\documentclass{lmcs}
\pdfoutput=1

\usepackage{lastpage}
\lmcsdoi{16}{1}{31}
\lmcsheading{}{\pageref{LastPage}}{}{}%
{Aug.~17,~2018}{Mar.~11,~2020}{}

\keywords{infinitary rewriting, confluence, normalisation, B{\"o}hm
  trees, coinduction}

\usepackage{amssymb}
\usepackage{latexsym} 
\usepackage{proof}
\usepackage[all]{xy}

\newcommand{\Nbb}{\ensuremath{\mathbb{N}}}
\newcommand{\Pow}{\ensuremath{\mathcal{P}}}

\newcommand{\la}{\ensuremath{\langle}}
\newcommand{\ra}{\ensuremath{\rangle}}

\newcommand{\Tc}{\ensuremath{{\mathcal T}}}
\newcommand{\Rc}{\ensuremath{{\mathcal R}}}
\newcommand{\Uc}{\ensuremath{{\mathcal U}}}
\newcommand{\Oc}{\ensuremath{{\mathcal O}}}
\newcommand{\Hc}{\ensuremath{{\mathcal H}}}

\newcommand{\Os}{\ensuremath{{\mathsf O}}}

\newcommand{\FV}{\ensuremath{\mathrm{FV}}}

\newcommand{\reduces}{\ensuremath{\to^*}}
\newcommand{\infred}{\ensuremath{\to^\infty}}
\newcommand{\To}{\ensuremath{\Rightarrow}}
\newcommand{\xto}{\xrightarrow}

\newcommand{\pos}[2]{\ensuremath{{#1}_{|#2}}}

\newcommand{\da}{{\downarrow}}
\newcommand{\ua}{{\uparrow}}

\newcommand{\cons}{\mathtt{cons}}
\newcommand{\even}{\mathtt{even}}
\newcommand{\tail}{\mathtt{tl}}
\newcommand{\head}{\mathtt{hd}}
\newcommand{\zip}{\mathtt{zip}}
\newcommand{\mer}{\mathtt{merge}}

\newcommand{\sfr}{\mathfrak{s}}
\newcommand{\dfr}{\mathfrak{d}}

\newcommand{\Type}{\mathrm{Type}}
\newcommand{\Prop}{\mathrm{Prop}}
\newcommand{\Set}{\mathrm{Set}}

\DeclareMathOperator{\Coloneqq}{{\,::=\,}}

\title{A new coinductive confluence proof for infinitary lambda calculus}

\author[{\L}.~Czajka]{{\L}ukasz Czajka}
\address{University of Copenhagen, Universitetsparken 5, 2100 Copenhagen, Denmark}
\email{luta@di.ku.dk}
\thanks{Supported by the European Union's Horizon 2020 research and
  innovation programme under the Marie Sk{\l}odowska-Curie grant
  agreement number~704111.}

\begin{document}

\begin{abstract}
  We present a new and formal coinductive proof of confluence and
  normalisation of B{\"o}hm reduction in infinitary
  lambda calculus. The proof is simpler than previous proofs of this
  result. The technique of the proof is new, i.e., it is not merely a
  coinductive reformulation of any earlier proofs. We formalised the
  proof in the Coq proof assistant.
\end{abstract}

\maketitle


\newcommand{\IND}{\mathrm{IND}}
\newcommand{\COIND}{\mathrm{COIND}}

\section{Introduction}

Infinitary lambda calculus is a generalisation of lambda calculus that
allows infinite lambda terms and transfinite reductions. This enables
the consideration of ``limits'' of terms under infinite reduction
sequences. For instance, for a term $M \equiv (\lambda m x . m m)
(\lambda m x . m m)$ we have
\[
M \to_\beta \lambda x . M \to_\beta \lambda x . \lambda x . M
\to_\beta \lambda x . \lambda x . \lambda x . M \to_\beta \ldots
\]
Intuitively, the ``value'' of~$M$ is an infinite term~$L$ satisfying
$L \equiv \lambda x . L$, where by~$\equiv$ we denote identity of
terms. In fact, $L$ is the normal form of~$M$ in infinitary
lambda calculus.

In~\cite{EndrullisPolonsky2011,EndrullisHansenHendriksPolonskySilva2018}
it is shown that infinitary reductions may be defined
coinductively. The standard non-coinductive definition makes explicit
mention of ordinals and limits in a certain metric
space~\cite{KennawayKlopSleepVries1997,KennawayVries2003,BarendregtKlop2009}. A
coinductive approach is better suited to formalisation in a
proof-assistant. Indeed, with relatively little effort we have
formalised our results in~Coq (see Section~\ref{sec_formalisation}).

In this paper we show confluence of infinitary lambda calculus, modulo
equivalence of so-called meaningless
terms~\cite{KennawayOostromVries1999}. We also show confluence and
normalisation of infinitary B{\"o}hm reduction over any set of
strongly meaningless terms. All these results have already been
obtained in~\cite{KennawayKlopSleepVries1997,KennawayOostromVries1999}
by a different and more complex proof method.

In a related conference paper~\cite{Czajka2014} we have shown
confluence of infinitary reduction modulo equivalence of root-active
subterms, and confluence of infinitary B{\"o}hm reduction over
root-active terms. The present paper is quite different
from~\cite{Czajka2014}. A new and simpler method is used. The proof
in~\cite{Czajka2014} follows the general strategy
of~\cite{KennawayKlopSleepVries1997}. There first confluence modulo
equivalence of root-active terms is shown, proving confluence of an
auxiliary $\epsilon$-calculus as an intermediate step. Then confluence
of B{\"ohm} reduction is derived from confluence modulo equivalence of
root-active terms. Here we first show that every term has a unique
normal form reachable by a special standard infinitary
$N_\Uc$-reduction. Then we use this result to derive confluence of
B{\"o}hm reduction, and from that confluence modulo equivalence of
meaningless terms. We do not use any $\epsilon$-calculus at all. See
the beginning of Section~\ref{sec_inflam_confluence} for a more
detailed discussion of our proof method.

\subsection{Related work}

Infinitary lambda calculus was introduced
in~\cite{KennawayKlopSleepVries1997,KennawayKlopSleepVries1995b}. Meaningless
terms were defined in~\cite{KennawayOostromVries1999}. The confluence
and normalisation results of this paper were already obtained
in~\cite{KennawayKlopSleepVries1997,KennawayOostromVries1999}, by a
different proof method. See
also~\cite{KennawayVries2003,BarendregtKlop2009,EndrullisHendriksKlop2012}
for an overview of various results in infinitary lambda calculus and
infinitary rewriting.

Joachimski in~\cite{Joachimski2004} gives a coinductive confluence
proof for infinitary lambda calculus, but Joachimski's notion of
reduction does not correspond to the standard notion of a strongly
convergent reduction. Essentially, it allows for infinitely many
parallel contractions in one step, but only finitely many reduction
steps. The coinductive definition of infinitary reductions capturing
strongly convergent reductions was introduced
in~\cite{EndrullisPolonsky2011}. Later~\cite{EndrullisHansenHendriksPolonskySilva2015,EndrullisHansenHendriksPolonskySilva2018}
generalised this to infinitary term rewriting
systems. In~\cite{Czajka2014} using the definition
from~\cite{EndrullisPolonsky2011}, confluence of infinitary
lambda calculus modulo equivalence of root-active subterms was shown
coinductively. The proof in~\cite{Czajka2014} follows the general
strategy
of~\cite{KennawayKlopSleepVries1997,KennawayKlopSleepVries1995b}. The
proof in the present paper bears some similarity to the proof of the
unique normal forms property of orthogonal~iTRSs
in~\cite{KlopVrijer2005}. It is also similar to the coinductive
confluence proof for nearly orthogonal infinitary term rewriting
systems in~\cite{Czajka2015a}, but there the ``standard'' reduction
employed is not unique and need not be normalising.

Confluence and normalisation results in infinitary rewriting and
infinitary lambda calculus have been generalised to the framework of
infinitary combinatory reduction
systems~\cite{KetemaSimonsen2009,KetemaSimonsen2010,KetemaSimonsen2011}.

There are three well-known variants of infinitary lambda calculus:
the~$\Lambda^{111}$, $\Lambda^{001}$ and~$\Lambda^{101}$
calculi~\cite{BarendregtKlop2009,EndrullisHendriksKlop2012,KennawayKlopSleepVries1997,KennawayKlopSleepVries1995b}. The
superscripts $111$, $001$, $101$ indicate the depth measure used:
$abc$ means that we shall add $a$/$b$/$c$ to the depth when going
down/left/right in the tree of the
lambda term~\cite[Definition~6]{KennawayKlopSleepVries1997}. In this
paper we are concerned only with a coinductive presentation of the
$\Lambda^{111}$-calculus.

In the $\Lambda^{001}$-calculus, after addition of appropriate
$\bot$-rules, every finite term has its B{\"o}hm
tree~\cite{KennawayKlopSleepVries1995b} as the normal
form. In~$\Lambda^{111}$ and~$\Lambda^{101}$, the normal forms are,
respectively, Berarducci trees and Levy-Longo
trees~\cite{KennawayKlopSleepVries1997,KennawayKlopSleepVries1995b,Berarducci1996b,Levy1975,Longo1983}. With
the addition of infinite $\eta$- or $\eta!$-reductions it is possible
to also capture, respectively, $\eta$-B{\"ohm} or
$\infty\eta$-B{\"o}hm trees as normal
forms~\cite{SeveriVries2002,SeveriVries2017}.

The addition of $\bot$-rules may be avoided by basing the definition
of infinitary terms on ideal completion. This line of work is pursued
in~\cite{Bahr2010,Bahr2014,Bahr2018}. Confluence of the resulting
calculi is shown, but the proof depends on the confluence results
of~\cite{KennawayKlopSleepVries1997}.

\section{Infinite terms and corecursion}\label{sec_corec}

In this section we define many-sorted infinitary terms. We also
explain and justify guarded corecursion using elementary notions. The
results of this section are well-known.

\begin{defi}\label{def_coterms}
  A \emph{many-sorted algebraic signature}
  $\Sigma=\la\Sigma_s,\Sigma_c\ra$ consists of a collection of
  \emph{sort symbols}~$\Sigma_s=\{s_i\}_{i\in I}$ and a collection of
  \emph{constructors} $\Sigma_c=\{c_j\}_{j\in J}$. Each
  constructor~$c$ has an associated \emph{type}
  $\tau(c)=(s_1,\ldots,s_n;s)$ where $s_1,\ldots,s_n,s\in\Sigma_s$. If
  $\tau(c)=(;s)$ then~$c$ is a \emph{constant} of sort~$s$. In what
  follows we use $\Sigma,\Sigma'$, etc., for many-sorted algebraic
  signatures, $s,s'$, etc., for sort symbols, and $f,g,c,d$, etc., for
  constructors.

  The set~$\Tc^\infty(\Sigma)$, or just~$\Tc(\Sigma)$, of
  \emph{infinitary terms over~$\Sigma$} is the set of all finite and
  infinite terms over~$\Sigma$, i.e., all finite and infinite labelled
  trees with labels of nodes specified by the constructors of~$\Sigma$
  such that the types of labels of nodes agree. More precisely, a
  term~$t$ over~$\Sigma$ is a partial function from~$\Nbb^*$
  to~$\Sigma_c$ satisfying:
  \begin{itemize}
  \item $t(\epsilon)\da$, and
  \item if $t(p) = c \in \Sigma_c$ with $\tau(c)=(s_1,\ldots,s_n;s)$
    then
    \begin{itemize}
    \item $t(pi)=d \in \Sigma_c$ with
      $\tau(d)=(s_1',\ldots,s_{m_i}';s_i)$ for $i < n$,
    \item $t(pi)\ua$ for $i \ge n$,
    \end{itemize}
  \item if $t(p)\ua$ then~$t(pi)\ua$ for every $i\in\Nbb$,
  \end{itemize}
  where $t(p)\ua$ means that~$t(p)$ is undefined, $t(p)\da$ means
  that~$t(p)$ is defined, and~$\epsilon \in \Nbb^*$ is the empty
  string. We use obvious notations for infinitary terms, e.g.,
  $f(g(t,s),c)$ when $c,f,g \in \Sigma_c$ and $t,s \in \Tc(\Sigma)$,
  and the types agree. We say that a term~$t$ \emph{is of sort~$s$} if
  $t(\epsilon)$ is a constructor of type $(s_1,\ldots,s_n;s)$ for some
  $s_1,\ldots,s_n\in\Sigma_s$. By~$\Tc_s(\Sigma)$ we denote the set of
  all terms of sort~$s$ from~$\Tc(\Sigma)$.
\end{defi}

\begin{exa}
  Let $A$ be a set. Let $\Sigma$ consist of two sorts~$\sfr$
  and~$\dfr$, one constructor~$\cons$ of type $(\dfr,\sfr;\sfr)$ and a
  distinct constant $a \in A$ of sort~$\dfr$ for each element
  of~$A$. Then~$\Tc_\sfr(\Sigma)$ is the set of streams over~$A$. We
  also write $\Tc_\sfr(\Sigma) = A^\omega$ and $\Tc_\dfr(\Sigma) =
  A$. Instead of $\cons(a,t)$ we usually write $a : t$, and we assume
  that~$:$ associates to the right, e.g., $x : y : t$ is $x : (y :
  t)$. We also use the notation $x : t$ to denote the application of
  the constructor for~$\cons$ to~$x$ and~$t$. We define the functions
  $\head : A^\omega \to A$ and $\tail : A^\omega \to A^\omega$ by
  \[
  \begin{array}{rcl}
    \head(a : t) &=& a \\
    \tail(a : t) &=& t
  \end{array}
  \]
  Specifications of many-sorted signatures may be conveniently given
  by coinductively interpreted grammars. For instance, the
  set~$A^\omega$ of streams over a set~$A$ could be specified by
  writing
  \[
  A^\omega \Coloneqq \cons(A, A^\omega).
  \]
  A more interesting example is that of finite and infinite binary trees
  with nodes labelled either with~$a$ or~$b$, and leaves labelled with
  one of the elements of a set~$V$:
  \[
  T \Coloneqq V \parallel a(T, T) \parallel b(T, T).
  \]
  As such specifications are not intended to be formal entities but
  only convenient notation for describing sets of infinitary terms, we
  will not define them precisely. It is always clear what many-sorted
  signature is meant.
\end{exa}

For the sake of brevity we often use $\Tc = \Tc(\Sigma)$ and $\Tc_s =
\Tc_s(\Sigma)$, i.e., we omit the signature~$\Sigma$ when clear from
the context or irrelevant.

\begin{defi}\label{def_guarded_corecursion}
  The class of \emph{constructor-guarded} functions is defined
  inductively as the class of all functions $h : \Tc_s^m \to \Tc_{s'}$
  (for arbitrary $m \in \Nbb$, $s,s' \in \Sigma_s$) such that there
  are a constructor~$c$ of type $(s_1,\ldots,s_{k};s')$ and functions
  $u_i : \Tc_s^m \to \Tc_{s_i}$ ($i=1,\ldots,k$) such that
  \[
  h(y_1,\ldots,y_m) =
  c(u_1(y_1,\ldots,y_m),\ldots,u_k(y_1,\ldots,y_m))
  \]
  for all $y_1,\ldots,y_m \in \Tc_s$, and for each $i=1,\ldots,k$ one
  of the following holds:
  \begin{itemize}
  \item $u_i$ is constructor-guarded, or
  \item $u_i$ is a constant function, or
  \item $u_i$ is a projection function, i.e., $s_i = s$ and there is
    $1\le j \le m$ with $u_i(y_1,\ldots,y_m) = y_j$ for all
    $y_1,\ldots,y_m \in \Tc_s$.
  \end{itemize}
  Let~$S$ be a set. A function $h : S \times \Tc_s^m \to \Tc_{s'}$ is
  constructor-guarded if for every $x \in S$ the function $h_x :
  \Tc_s^m \to \Tc_{s'}$ defined by $h_x(y_1,\ldots,y_m) =
  h(x,y_1,\ldots,y_m)$ is constructor-guarded. A function $f : S \to
  \Tc_s$ is defined by \emph{guarded corecursion} from $h : S \times
  \Tc_s^m \to \Tc_s$ and $g_i : S \to S$ ($i=1,\ldots,m$) if~$h$ is
  constructor-guarded and~$f$ satisfies
  \[
  f(x) = h(x, f(g_1(x)), \ldots, f(g_m(x)))
  \]
  for all $x \in S$.
\end{defi}

The following theorem is folklore in the coalgebra community. We
sketch an elementary proof. In fact, each set of many-sorted
infinitary terms is a final coalgebra of an appropriate
set-functor. Then Theorem~\ref{thm_corecursion} follows from more
general principles. We prefer to avoid coalgebraic terminology, as it
is simply not necessary for the purposes of the present paper. See
e.g.~\cite{JacobsRutten2011,Rutten2000} for a more general coalgebraic
explanation of corecursion.

\begin{thm}\label{thm_corecursion}
  For any constructor-guarded function $h : S \times \Tc_s^m \to
  \Tc_{s}$ and any $g_i : S \to S$ ($i=1,\ldots,m$), there exists a
  unique function $f : S \to \Tc_s$ defined by guarded corecursion
  from~$h$ and~$g_1,\ldots,g_m$.
\end{thm}

\begin{proof}
  Let $f_0 : S \to \Tc_s$ be an arbitrary function. Define~$f_{n+1}$
  for $n \in \Nbb$ by $f_{n+1}(x) = h(x, f_n(g_1(x)), \ldots,
  f_n(g_m(x)))$. Using the fact that~$h$ is constructor-guarded, one
  shows by induction on~$n$ that:
  \[
    \text{$f_{n+1}(x)(p) = f_n(x)(p)$ for $x \in S$ and $p
    \in \Nbb^*$ with $|p|<n$}
    \tag{$\star$}
  \]
  where~$|p|$ denotes the length of~$p$. Indeed, the base is
  obvious. We show the inductive step. Let $x \in S$. Because~$h$ is
  constructor-guarded, we have for instance
  \[
  f_{n+2}(x) = h(x, f_{n+1}(g_1(x)), \ldots, f_{n+1}(g_m(x))) =
  c_1(c_2, c_3(w, f_{n+1}(g_1(x))))
  \]
  Let $p \in \Nbb^*$ with $|p| \le n$. The only interesting case is
  when $p=11p'$, i.e., when~$p$ points inside~$f_{n+1}(g_1(x))$. But
  then $|p'| < |p| \le n$, so by the inductive hypothesis
  $f_{n+1}(g_1(x))(p') = f_n(g_1(x))(p')$. Thus $f_{n+2}(x)(p) =
  f_{n+1}(g_1(x))(p') = f_n(g_1(x))(p') = f_{n+1}(x)(p)$.

  Now we define $f : S \to \Tc_s$ by
  \[
  f(x)(p) = f_{|p|+1}(x)(p)
  \]
  for $x \in S$, $p \in \Nbb^*$. Using~$(\star)$ it is routine to
  check that~$f(x)$ is a well-defined infinitary term for each $x \in
  S$. To show that~$f : S \to \Tc_s$ is defined by guarded corecursion
  from~$h$ and~$g_1,\ldots,g_m$, using~$(\star)$ one shows by
  induction on the length of~$p \in \Nbb^*$ that for any $x \in S$:
  \[
  f(x)(p) = h(x, f(g_1(x)),\ldots,f(g_m(x)))(p).
  \]
  To prove that~$f$ is unique it suffices to show that it does not
  depend on~$f_0$. For this purpose, using~$(\star)$ one shows by
  induction on the length of~$p \in \Nbb^*$ that~$f(x)(p)$ does not
  depend on~$f_0$ for any $x \in S$.
\end{proof}

We shall often use the above theorem implicitly, just mentioning that
some equations define a function by guarded corecursion.

\begin{exa}\label{ex_corec}
  Consider the equation
  \[
  \even(x : y : t) = x : \even(t)
  \]
  It may be rewritten as
  \[
  \even(t) = \head(t) : \even(\tail(\tail(t)))
  \]
  So $\even : A^\omega \to A^\omega$ is defined by guarded corecursion
  from $h : A^\omega \times A^\omega \to A^\omega$ given by
  \[
  h(t,t') = \head(t) : t'
  \]
  and $g : A^\omega \to A^\omega$ given by
  \[
  g(t) = \tail(\tail(t))
  \]
  By Theorem~\ref{thm_corecursion} there is a unique function $\even :
  A^\omega \to A^\omega$ satisfying the original equation.

  Another example of a function defined by guarded corecursion is
  $\zip : A^\omega \times A^\omega \to A^\omega$:
  \[
  \zip(x : t, s) = x : \zip(s, t)
  \]
  The following function $\mer : \Nbb^\omega \times \Nbb^\omega \to
  \Nbb^\omega$ is also defined by guarded corecursion:
  \[
  \mer(x : t_1, y : t_2) =
  \left\{
    \begin{array}{cl}
      x : \mer(t_1, y : t_2) & \text{ if } x \le y \\
      y : \mer(x : t_1, t_2) & \text{ otherwise }
    \end{array}
  \right.
 \]
\end{exa}

\section{Coinduction}\label{sec_coind}

In this section\footnote{This section is largely based
  on~\cite[Section~2]{Czajka2015a}.} we give a brief explanation of
coinduction as it is used in the present paper. Our presentation of
coinductive proofs is similar to
e.g.~\cite{EndrullisPolonsky2011,BezemNakataUustalu2012,NakataUustalu2010,LeroyGrall2009,KozenSilva2017}.

There are many ways in which our coinductive proofs could be
justified. Since we formalised our main results (see
Section~\ref{sec_formalisation}), the proofs may be understood as a
paper presentation of formal Coq proofs. They can also be justified by
appealing to one of a number of established coinduction
principles. With enough patience one could, in principle, reformulate
all proofs to directly employ the usual coinduction principle in set
theory based on the Knaster-Tarski fixpoint
theorem~\cite{Sangiorgi2012}. One could probably also use the
coinduction principle from~\cite{KozenSilva2017}. Finally, one may
justify our proofs by indicating how to interpret them in ordinary set
theory, which is what we do in this section.

The purpose of this section is to explain how to justify our proofs by
reducing coinduction to transfinite induction. The present section
does not provide a formal coinduction proof principle as such, but
only indicates how one could elaborate the proofs so as to eliminate
the use of coinduction. Naturally, such an elaboration would introduce
some tedious details. The point is that all these details are
essentially the same for each coinductive proof. The advantage of
using coinduction is that the details need not be provided each
time. A similar elaboration could be done to directly employ any of a
number of established coinduction principles, but as far as we know
elaborating the proofs in the way explained here requires the least
amount of effort in comparison to reformulating them to directly
employ an established coinduction principle. In fact, we do not wish
to explicitly commit to any single formal proof principle, because we
do not think that choosing a specific principle has an essential
impact on the content of our proofs, except by making it more or less
straightforward to translate the proofs into a form which uses the
principle directly.

A reader not satisfied with the level of rigour of the explanations of
coinduction below is referred to our formalisation (see
Section~\ref{sec_formalisation}). The present section provides one way
in which our proofs can be understood and verified without resorting
to a formalisation. To make the observations of this section
completely precise and general one would need to introduce formal
notions of ``proof'' and ``statement''. In other words, one would need
to formulate a system of logic with a capacity for coinductive
proofs. We do not want to do this here, because this paper is about a
coinductive confluence proof for infinitary lambda calculus, not about
foundations of coinduction. It would require some work, but should not
be too difficult, to create a formal system based on the present
section in which our coinductive proofs could be interpreted
reasonably directly. We defer this to future work. The status of the
present section is that of a ``meta-explanation'', analogous to an
explanation of how, e.g., the informal presentations of inductive
constructions found in the literature may be encoded in ZFC set
theory.

\begin{exa}\label{ex_1}
  Let~$T$ be the set of all finite and infinite terms defined
  coinductively by
  \[
  T \Coloneqq V \parallel A(T) \parallel B(T, T)
  \]
  where~$V$ is a countable set of variables, and~$A$, $B$ are
  constructors. By $x,y,\ldots$ we denote variables, and by
  $t,s,\ldots$ we denote elements of~$T$. Define a binary
  relation~$\to$ on~$T$ coinductively by the following rules.
  \[
  \infer=[(1)]{x \to x}{} \quad
  \infer=[(2)]{A(t) \to A(t')}{t \to t'} \quad
  \infer=[(3)]{B(s,t) \to B(s',t')}{s \to s' & t \to t'} \quad
  \infer=[(4)]{A(t) \to B(t',t')}{t\to t'}
  \]

  Formally, the relation~${\to}$ is the greatest fixpoint of a
  monotone function
  \[
  F : \Pow(T \times T) \to \Pow(T \times T)
  \]
  defined by
  \[
  F(R) = \left\{ \la t_1, t_2 \ra \mid \exists_{x \in V}(t_1 \equiv
    t_2 \equiv x) \lor \exists_{t,t'\in T}(t_1 \equiv A(t) \land t_2
    \equiv A(t') \land R(t,t')) \lor \ldots \right\}.
  \]

  Alternatively, using the Knaster-Tarski fixpoint theorem, the
  relation~$\to$ may be characterised as the greatest binary relation
  on~$T$ (i.e. the greatest subset of $T\times T$ w.r.t.~set
  inclusion) such that ${\to} \subseteq F({\to})$, i.e., such that for
  every $t_1,t_2 \in T$ with $t_1 \to t_2$ one of the following holds:
  \begin{enumerate}
  \item $t_1 \equiv t_2 \equiv x$ for some variable $x \in V$,
  \item $t_1 \equiv A(t)$, $t_2 \equiv A(t')$ with $t \to t'$,
  \item $t_1 \equiv B(s,t)$, $t_2 \equiv B(s',t')$ with $s \to s'$ and
    $t \to t'$,
  \item $t_1 \equiv A(t)$, $t_2 \equiv B(t',t')$ with $t \to t'$.
  \end{enumerate}

  Yet another way to think about~$\to$ is that $t_1 \to t_2$ holds if
  and only if there exists a \emph{potentially infinite} derivation
  tree of $t_1 \to t_2$ built using the rules~$(1)-(4)$.

  The rules~$(1)-(4)$ could also be interpreted inductively to yield
  the least fixpoint of~$F$. This is the conventional interpretation,
  and it is indicated with a single line in each rule separating
  premises from the conclusion. A coinductive interpretation is
  indicated with double lines.

  The greatest fixpoint~$\to$ of~$F$ may be obtained by transfinitely
  iterating~$F$ starting with~$T \times T$. More precisely, define an
  ordinal-indexed sequence~$(\to^\gamma)_\gamma$ by:
  \begin{itemize}
  \item $\to^0 = T \times T$,
  \item $\to^{\gamma+1} = F(\to^\gamma)$,
  \item $\to^\delta = \bigcap_{\gamma<\delta} \to^\gamma$ for a limit
    ordinal~$\delta$.
  \end{itemize}
  Then there exists an ordinal~$\zeta$ such that ${\to} =
  {\to^\zeta}$. The least such ordinal is called the \emph{closure
    ordinal}. Note also that ${\to^\gamma} \subseteq {\to^\delta}$ for $\gamma \ge \delta$
  (we often use this fact implicitly). See
  e.g.~\cite[Chapter~8]{DaveyPriestley2002}. The relation~$\to^\gamma$ is
  called the \emph{$\gamma$-approximant}. Note that the $\gamma$-approximants
  depend on a particular definition of~$\to$ (i.e.~on the
  function~$F$), not solely on the relation~$\to$ itself. We
  use~$R^\gamma$ for the $\gamma$-approximant of a relation~$R$.

  It is instructive to note that the coinductive rules for~$\to$ may
  also be interpreted as giving rules for the $\gamma+1$-approximants,
  for any ordinal~$\gamma$.
  \[
  \infer[(1)]{x \to^{\gamma+1} x}{}\quad
  \infer[(2)]{A(t) \to^{\gamma+1} A(t')}{
    t \to^\gamma t'}\quad
  \infer[(3)]{B(s,t) \to^{\gamma+1} B(s',t')}{
    s \to^\gamma s' & t \to^\gamma t'}\quad
  \infer[(4)]{A(t) \to^{\gamma+1} B(t',t')}{
    t\to^\gamma t'}
  \]

  Usually, the closure ordinal for the definition of a coinductive
  relation is~$\omega$, as is the case with all coinductive
  definitions appearing in this paper. In general, however, it is not
  difficult to come up with a coinductive definition whose closure
  ordinal is greater than~$\omega$. For instance, consider the
  relation $R \subseteq \Nbb \cup \{\infty\}$ defined coinductively by
  the following two rules.
  \[
  \infer={R(n+1)}{R(n) & n \in \Nbb} \quad\quad
  \infer={R(\infty)}{\exists n \in \Nbb . R(n)}
  \]
  We have $R = \emptyset$,
  $R^n = \{m \in \Nbb \mid m \ge n\} \cup \{\infty\}$ for
  $n \in \Nbb$, $R^\omega = \{\infty\}$, and only
  $R^{\omega+1}=\emptyset$. Thus the closure ordinal of this
  definition is $\omega+1$.
\end{exa}

In this paper we are interested in proving by coinduction statements
of the form $\psi(R_1,\ldots,R_m)$ where
\[
\psi(X_1,\ldots,X_m) \equiv \forall x_1 \ldots x_n . \varphi(\vec{x})
\to X_1(g_1(\vec{x}),\ldots,g_k(\vec{x})) \land \ldots \land
X_m(g_1(\vec{x}),\ldots,g_k(\vec{x})).
\]
and $R_1,\ldots,R_m$ are coinductive relations on~$T$, i.e, relations
defined as the greatest fixpoints of some monotone functions on the
powerset of an appropriate cartesian product of~$T$, and
$\psi(R_1,\ldots,R_m)$ is~$\psi(X_1,\ldots,X_m)$ with~$R_i$
substituted for~$X_i$. Statements with an existential quantifier may
be reduced to statements of this form by skolemising, as explained in
Example~\ref{ex_skolem} below.

Here $X_1,\ldots,X_m$ are meta-variables for which relations on~$T$
may be substituted. In the statement~$\varphi(\vec{x})$ only
$x_1,\ldots,x_n$ occur free. The meta-variables $X_1,\ldots,X_m$
\emph{are not allowed to occur} in~$\varphi(\vec{x})$. In general, we
abbreviate $x_1,\ldots,x_n$ with~$\vec{x}$. The
symbols~$g_1,\ldots,g_k$ denote some functions of~$\vec{x}$.

To prove~$\psi(R_1,\ldots,R_m)$ it suffices to show by transfinite
induction that $\psi(R_1^\gamma,\ldots,R_m^\gamma)$ holds for each
ordinal~$\gamma \le \zeta$, where~$R_i^\gamma$ is the
$\gamma$-approximant of~$R_i$. It is an easy exercise to check that
because of the special form of~$\psi$ (in particular because~$\varphi$
does not contain~$X_1,\ldots,X_m$) and the fact that each~$R_i^0$ is
the full relation, the base case~$\gamma=0$ and the case of~$\gamma$ a
limit ordinal hold. They hold for \emph{any}~$\psi$ of the above form,
\emph{irrespective} of $\varphi,R_1,\ldots,R_m$. Note
that~$\varphi(\vec{x})$ is the same in
all~$\psi(R_1^\gamma,\ldots,R_m^\gamma)$ for any~$\gamma$, i.e., it
does not refer to the $\gamma$-approximants or the
ordinal~$\gamma$. Hence it remains to show the inductive step
for~$\gamma$ a successor ordinal. It turns out that a coinductive
proof of~$\psi$ may be interpreted as a proof of this inductive step
for a successor ordinal, with the ordinals left implicit and the
phrase ``coinductive hypothesis'' used instead of ``inductive
hypothesis''.

\begin{exa}
  On terms from~$T$ (see Example~\ref{ex_1}) we define the operation
  of substitution by guarded corecursion.
  \[
  \begin{array}{rclcrcl}
    y[t/x] &=& y \quad\text{ if } x \ne y &\quad&
    (A(s))[t/x] &=& A(s[t/x]) \\
    x[t/x] &=& t &\quad&
    (B(s_1,s_2))[t/x] &=& B(s_1[t/x],s_2[t/x])
  \end{array}
  \]
  We show by coinduction: if $s \to s'$ and $t \to t'$ then $s[t/x]
  \to s'[t'/x]$, where~$\to$ is the relation from
  Example~\ref{ex_1}. Formally, the statement we show by transfinite
  induction on~$\gamma \le \zeta$ is: for $s,s',t,t' \in T$, if $s \to
  s'$ and $t \to t'$ then $s[t/x] \to^\gamma s'[t'/x]$. For
  illustrative purposes, we indicate the $\gamma$-approximants with
  appropriate ordinal superscripts, but it is customary to omit these
  superscripts.

  Let us proceed with the proof. The proof is by coinduction with case
  analysis on $s \to s'$. If $s \equiv s' \equiv y$ with $y \ne x$,
  then $s[t/x] \equiv y \equiv s'[t'/x]$. If $s \equiv s' \equiv x$
  then $s[t/x] \equiv t \to^{\gamma+1} t' \equiv s'[t'/x]$ (note that
  ${\to} \equiv {\to^\zeta} \subseteq {\to^{\gamma+1}}$). If $s \equiv
  A(s_1)$, $s' \equiv A(s_1')$ and $s_1 \to s_1'$, then $s_1[t/x]
  \to^\gamma s_1'[t'/x]$ by the coinductive hypothesis. Thus $s[t/x]
  \equiv A(s_1[t/x]) \to^{\gamma+1} A(s_1'[t'/x]) \equiv s'[t'/x]$ by
  rule~$(2)$. If $s \equiv B(s_1,s_2)$, $s' \equiv B(s_1',s_2')$ then
  the proof is analogous. If $s \equiv A(s_1)$, $s' \equiv
  B(s_1',s_1')$ and $s_1 \to s_1'$, then the proof is also
  similar. Indeed, by the coinductive hypothesis we have $s_1[t/x]
  \to^\gamma s_1'[t'/x]$, so $s[t/x] \equiv A(s_1[t/x]) \to^{\gamma+1}
  B(s_1'[t'/x],s_1'[t'/x]) \equiv s'[t'/x]$ by rule~$(4)$.
\end{exa}

With the following example we explain how our proofs of existential
statements should be interpreted.

\begin{exa}\label{ex_skolem}
  Let~$T$ and~$\to$ be as in Example~\ref{ex_1}. We want to show: for
  all $s,t,t' \in T$, if $s \to t$ and $s \to t'$ then there exists
  $s' \in T$ with $t \to s'$ and $t' \to s'$. The idea is to skolemise
  this statement. So we need to find a Skolem function $f : T^3 \to T$
  which will allow us to prove the Skolem normal form:
  \[
    \text{if $s \to t$ and $s \to t'$ then $t \to f(s,t,t')$
      and $t' \to f(s,t,t')$.}
    \tag{$\star$}
  \]
  The rules for~$\to$ suggest a definition of~$f$:
  \[
  \begin{array}{rcl}
    f(x, x, x) &=& x \\
    f(A(s), A(t), A(t')) &=& A(f(s,t,t')) \\
    f(A(s),A(t),B(t',t')) &=& B(f(s,t,t'),f(s,t,t')) \\
    f(A(s),B(t,t),A(t')) &=& B(f(s,t,t'),f(s,t,t')) \\
    f(A(s),B(t,t),B(t',t')) &=& B(f(s,t,t'),f(s,t,t')) \\
    f(B(s_1,s_2), B(t_1,t_2), B(t_1',t_2')) &=&
    B(f(s_1,t_1,t_1'),f(s_2,t_2,t_2')) \\
    f(s, t, t') &=& \text{some fixed term if none of the above matches}
  \end{array}
  \]
  This is a definition by guarded corecursion, so there exists a
  unique function $f : T^3 \to T$ satisfying the above equations. The
  last case in the above definition of~$f$ corresponds to the case in
  Definition~\ref{def_guarded_corecursion} where all~$u_i$ are
  constant functions. Note that any fixed term has a fixed constructor
  (in the sense of Definition~\ref{def_guarded_corecursion}) at the
  root. In the sense of Definition~\ref{def_guarded_corecursion} also
  the elements of~$V$ are nullary constructors.

  We now proceed with a coinductive proof of~$(\star)$. Assume $s \to
  t$ and $s \to t'$. If $s \equiv t \equiv t' \equiv x$ then
  $f(s,t,t') \equiv x$, and $x \to x$ by rule~$(1)$. If $s \equiv A(s_1)$,
  $t \equiv A(t_1)$ and $t' \equiv A(t_1')$ with $s_1 \to t_1$ and
  $s_1 \to t_1'$, then by the coinductive hypothesis $t_1 \to
  f(s_1,t_1,t_1')$ and $t_1' \to f(s_1,t_1,t_1')$. We have $f(s,t,t')
  \equiv A(f(s_1,t_1,t_1'))$. Hence $t \equiv A(t_1) \to f(s,t,t')$
  and $t \equiv A(t_1') \to f(s,t,t')$, by rule~$(2)$. If $s \equiv
  B(s_1,s_2)$, $t \equiv B(t_1,t_2)$ and $t' \equiv B(t_1',t_2')$,
  with $s_1 \to t_1$, $s_1 \to t_1'$, $s_2 \to t_2$ and $s_2 \to
  t_2'$, then by the coinductive hypothesis we have $t_1 \to
  f(s_1,t_1,t_1')$, $t_1' \to f(s_1,t_1,t_1')$, $t_2 \to
  f(s_2,t_2,t_2')$ and $t_2' \to f(s_2,t_2,t_2')$. Hence $t \equiv
  B(t_1,t_2) \to B(f(s_1,t_1,t_1'),f(s_2,t_2,t_2')) \equiv f(s,t,t')$
  by rule~$(3)$. Analogously, $t' \to f(s,t,t')$ by rule~$(3)$. Other
  cases are similar.

  Usually, it is inconvenient to invent the Skolem function
  beforehand, because the definition of the Skolem function and the
  coinductive proof of the Skolem normal form are typically
  interdependent. Therefore, we adopt an informal style of doing a
  proof by coinduction of a statement
  \[
  \begin{array}{rcl}
    \psi(R_1,\ldots,R_m) &=& \forall_{x_1, \ldots, x_n \in T}
    \,.\, \varphi(\vec{x}) \to \\ &&\quad \exists_{y \in T}
    . R_1(g_1(\vec{x}),\ldots,g_k(\vec{x}), y) \land \ldots \land R_m(g_1(\vec{x}),\ldots,g_k(\vec{x}),y)
  \end{array}
  \]
  with an existential quantifier. We intertwine the corecursive
  definition of the Skolem function~$f$ with a coinductive proof of
  the Skolem normal form
  \[
  \begin{array}{l}
    \forall_{x_1, \ldots, x_n \in T} \,.\, \varphi(\vec{x}) \to
    \\ \quad\quad
    R_1(g_1(\vec{x}),\ldots,g_k(\vec{x}),f(\vec{x}))
    \land \ldots \land R_m(g_1(\vec{x}),\ldots,g_k(\vec{x}),f(\vec{x}))
  \end{array}
  \]
  We proceed as if the coinductive hypothesis
  was~$\psi(R_1^\gamma,\ldots,R_m^\gamma)$ (it really is the Skolem
  normal form). Each element obtained from the existential quantifier
  in the coinductive hypothesis is interpreted as a corecursive
  invocation of the Skolem function. When later we exhibit an element
  to show the existential subformula
  of~$\psi(R_1^{\gamma+1},\ldots,R_m^{\gamma+1})$, we interpret this
  as the definition of the Skolem function in the case specified by
  the assumptions currently active in the proof. Note that this
  exhibited element may (or may not) depend on some elements obtained
  from the existential quantifier in the coinductive hypothesis, i.e.,
  the definition of the Skolem function may involve corecursive
  invocations.

  To illustrate our style of doing coinductive proofs of statements
  with an existential quantifier, we redo the proof done above. For
  illustrative purposes, we indicate the arguments of the Skolem
  function, i.e., we write~$s'_{s,t,t'}$ in place
  of~$f(s,t,t')$. These subscripts $s,t,t'$ are normally omitted.

  We show by coinduction that if $s \to t$ and $s \to t'$ then there
  exists $s' \in T$ with $t \to s'$ and $t' \to s'$. Assume $s \to t$
  and $s \to t'$. If $s \equiv t \equiv t' \equiv x$ then take
  $s'_{x,x,x} \equiv x$. If $s \equiv A(s_1)$, $t \equiv A(t_1)$ and
  $t' \equiv A(t_1')$ with $s_1 \to t_1$ and $s_1 \to t_1'$, then by
  the coinductive hypothesis we obtain~$s'_{s_1,t_1,t_1'}$ with
  $t_1 \to s'_{s_1,t_1,t_1'}$ and $t_1' \to s'_{s_1,t_1,t_1'}$. More
  precisely: by corecursively applying the Skolem function to
  $s_1,t_1,t_1'$ we obtain~$s'_{s_1,t_1,t_1'}$, and by the coinductive
  hypothesis we have $t_1 \to s'_{s_1,t_1,t_1'}$ and
  $t_1' \to s'_{s_1,t_1,t_1'}$. Hence
  $t \equiv A(t_1) \to A(s'_{s_1,t_1,t_1'})$ and
  $t \equiv A(t_1') \to A(s'_{s_1,t_1,t_1'})$, by rule~$(2)$. Thus we
  may take $s'_{s,t,t'} \equiv A(s'_{s_1,t_1,t_1'})$. If
  $s \equiv B(s_1,s_2)$, $t \equiv B(t_1,t_2)$ and
  $t' \equiv B(t_1',t_2')$, with $s_1 \to t_1$, $s_1 \to t_1'$,
  $s_2 \to t_2$ and $s_2 \to t_2'$, then by the coinductive hypothesis
  we obtain~$s'_{s_1,t_1,t_1'}$ and~$s'_{s_2,t_2,t_2'}$ with
  $t_1 \to s'_{s_1,t_1,t_1'}$, $t_1' \to s'_{s_1,t_1,t_1'}$,
  $t_2 \to s'_{s_2,t_2,t_2'}$ and $t_2' \to s'_{s_2,t_2,t_2'}$. Hence
  $t \equiv B(t_1,t_2) \to B(s'_{s_1,t_1,t_1'},s'_{s_2,t_2,t_2'})$ by
  rule~$(3)$. Analogously,
  $t' \to B(s'_{s_1,t_1,t_1'},s'_{s_2,t_2,t_2'})$ by rule~$(3)$. Thus
  we may take
  $s'_{s,t,t'} \equiv B(s'_{s_1,t_1,t_1'},s'_{s_2,t_2,t_2'})$. Other
  cases are similar.

  It is clear that the above informal proof, when interpreted in the
  way outlined before, implicitly defines the Skolem function~$f$. It
  should be kept in mind that in every case the definition of the
  Skolem function needs to be guarded. We do not explicitly mention
  this each time, but verifying this is part of verifying the proof.
\end{exa}

When doing proofs by coinduction the following criteria need to be
kept in mind in order to be able to justify the proofs according to
the above explanations.
\begin{itemize}
\item When we conclude from the coinductive hypothesis that some
  relation~$R(t_1,\ldots,t_n)$ holds, this really means that only its
  approximant~$R^\gamma(t_1,\ldots,t_n)$ holds. Usually, we need to
  infer that the next approximant~$R^{\gamma+1}(s_1,\ldots,s_n)$ holds
  (for some other elements~$s_1,\ldots,s_n$) by
  using~$R^\gamma(t_1,\ldots,t_n)$ as a premise of an appropriate
  rule. But we cannot, e.g., inspect (do case reasoning
  on)~$R^\gamma(t_1,\ldots,t_n)$, use it in any lemmas, or otherwise
  treat it as~$R(t_1,\ldots,t_n)$.
\item An element~$e$ obtained from an existential quantifier in the
  coinductive hypothesis is not really the element itself, but a
  corecursive invocation of the implicit Skolem function. Usually, we
  need to put it inside some constructor~$c$, e.g.~producing~$c(e)$,
  and then exhibit~$c(e)$ in the proof of an existential
  statement. Applying at least one constructor to~$e$ is necessary to
  ensure guardedness of the implicit Skolem function. But we cannot,
  e.g., inspect~$e$, apply some previously defined functions to it, or
  otherwise treat it as if it was really given to us.
\item In the proofs of existential statements, the implicit Skolem
  function cannot depend on the ordinal~$\gamma$. However, this is the
  case as long as we do not violate the first point, because if the
  ordinals are never mentioned and we do not inspect the approximants
  obtained from the coinductive hypothesis, then there is no way in
  which we could possibly introduce a dependency on~$\gamma$.
\end{itemize}

Equality on infinitary terms may be characterised coinductively.

\begin{defi}\label{def_bisimilarity}
  Let~$\Sigma$ be a many-sorted algebraic signature, as in
  Definition~\ref{def_coterms}. Let $\Tc = \Tc(\Sigma)$. Define
  on~$\Tc$ a binary relation~${=}$ of \emph{bisimilarity} by the
  coinductive rules
  \[
  \infer={f(t_1,\ldots,t_n) = f(s_1,\ldots,s_n)}{t_1 = s_1 & \ldots &
    t_n = s_n}
  \]
  for each constructor $f \in \Sigma_c$.
\end{defi}

It is intuitively obvious that on infinitary terms bisimilary is the
same as identity. The following easy (and well-known) proposition
makes this precise.

\begin{prop}\label{prop_bisimilarity}
  For $t,s \in \Tc$ we have: $t = s$ iff $t \equiv s$.
\end{prop}

\begin{proof}
  Recall that each term is formally a partial function from~$\Nbb^*$
  to~$\Sigma_c$. We write $t(p) \approx s(p)$ if either both
  $t(p),s(p)$ are defined and equal, or both are undefined.

  Assume $t = s$. It suffices to show by induction of the length of $p
  \in \Nbb^*$ that $\pos{t}{p} = \pos{s}{p}$ or $t(p)\ua,s(p)\ua$,
  where by~$\pos{t}{p}$ we denote the subterm of~$t$ at
  position~$p$. For $p = \epsilon$ this is obvious. Assume $p =
  p'j$. By the inductive hypothesis, $\pos{t}{p'} = \pos{s}{p'}$ or
  $t(p')\ua, s(p')\ua$. If $\pos{t}{p'} = \pos{s}{p'}$ then
  $\pos{t}{p'} \equiv f(t_0,\ldots,t_n)$ and $\pos{s}{p'} \equiv
  f(s_0,\ldots,s_n)$ for some $f \in \Sigma_c$ with $t_i = s_i$ for
  $i=0,\ldots,n$. If $0 \le j \le n$ then $\pos{t}{p} \equiv t_j = s_j
  = \pos{s}{p}$. Otherwise, if $j > n$ or if $t(p')\ua,s(p')\ua$, then
  $t(p)\ua,s(p)\ua$ by the definition of infinitary terms.

  For the other direction, we show by coinduction that for any $t \in
  \Tc$ we have $t = t$. If $t \in \Tc$ then $t \equiv
  f(t_1,\ldots,t_n)$ for some $f \in \Sigma_c$. By the coinductive
  hypothesis we obtain $t_i = t_i$ for $i=1,\ldots,n$. Hence $t = t$
  by the rule for~$f$.
\end{proof}

For infinitary terms $t,s\in \Tc$, we shall therefore use the
notations $t = s$ and $t \equiv s$ interchangeably, employing
Proposition~\ref{prop_bisimilarity} implicitly. In particular, the
above coinductive characterisation of term equality is used implicitly
in the proof of Lemma~\ref{lem_N_confluent}.

\begin{exa}
  Recall the coinductive definitions of~$\zip$ and~$\even$ from
  Example~\ref{ex_corec}.
  \[
  \begin{array}{rcl}
    \even(x : y : t) &=& x : \even(t) \\
    \zip(x : t, s) &=& x : \zip(s, t)
  \end{array}
  \]
  By coinduction we show
  \[
  \zip(\even(t),\even(\tail(t))) = t
  \]
  for any stream $t \in A^\omega$. Let $t \in A^\omega$. Then $t = x :
  y : s$ for some $x, y \in A$ and $s \in A^\omega$. We have
  \[
  \begin{array}{rcl}
    \zip(\even(t),\even(\tail(t))) &=& \zip(\even(x : y : s), \even(y
    : s)) \\
    &=& \zip(x : \even(s), \even(y : s)) \\
    &=& x : \zip(\even(y : s), \even(s)) \\
    &=& x : y : s \quad\text{ (by~CH) }\\
    &=& t
  \end{array}
  \]
  In the equality marked with~(by~CH) we use the coinductive
  hypothesis, and implicitly a bisimilarity rule from
  Definition~\ref{def_bisimilarity}.
\end{exa}

The above explanation of coinduction is generalised and elaborated in
much more detail in~\cite{Czajka2015}. Also~\cite{KozenSilva2017} may
be helpful as it gives many examples of coinductive proofs written in
a style similar to the one used here. The book~\cite{Sangiorgi2012} is
an elementary introduction to coinduction and bisimulation (but the
proofs there are presented in a different way than here, not referring
to the coinductive hypothesis but explicitly constructing a
backward-closed set). The
chapters~\cite{BertotCasteran2004Chapter13,Chlipala2013Chapter5}
explain coinduction in~Coq from a practical viewpoint. A reader
interested in foundational matters should also
consult~\cite{JacobsRutten2011,Rutten2000} which deal with the
coalgebraic approach to coinduction.

In the rest of this paper we shall freely use coinduction, giving
routine coinductive proofs in as much (or as little) detail as it is
customary with inductive proofs of analogous difficulty.

\section{Definitions and basic properties}\label{sec_inflam_intro}

In this section we define infinitary lambda terms and the various
notions of infinitary reductions.

\newcommand{\free}{\mathtt{free}}

\begin{defi}\label{def_infinitary_lambda_terms}
  The set of \emph{infinitary lambda terms} is defined coinductively:
  \[
  \begin{array}{rcl}
    \Lambda^\infty &::=& C \parallel V \parallel \Lambda^\infty\Lambda^\infty \parallel \lambda V . \Lambda^\infty
  \end{array}
  \]
  where~$V$ is an infinite set of \emph{variables} and~$C$ is a set of
  \emph{constants} such that $V \cap C = \emptyset$. An \emph{atom} is
  a variable or a constant. We use the symbols $x,y,z,\ldots$ for
  variables, and $c,c',c_1,\ldots$ for constants, and
  $a,a',a_1,\ldots$ for atoms, and $t,s,\ldots$ for terms. By~$\FV(t)$
  we denote the set of variables occurring free in~$t$. Formally,
  $\FV(t)$ could be defined using coinduction.

  We define substitution by guarded corecursion.
  \[
  \begin{array}{rcl}
    x[t/x] &=& t \\
    a[t/x] &=& a \quad\text{if } a \ne x \\
    (t_1t_2)[t/x] &=& (t_1[t/x]) (t_2[t/x]) \\
    (\lambda y . s)[t/x] &=& \lambda y . s[t/x] \quad\text{if } y
                             \notin \FV(t,x)
  \end{array}
  \]
\end{defi}

In our formalisation we use a de Bruijn representation of infinitary
lambda terms, defined analogously to the de Bruijn representation of
finite lambda terms~\cite{Bruijn1972}. Hence, infinitary lambda terms
here may be understood as a human-readable presentation of infinitary
lambda terms based on de Bruijn indices. Strictly speaking, also the
definition of substitution above is not completely precise, because it
implicitly treats lambda terms up to renaming of bound variables and
we have not given a precise definition of free variables. The
definition of substitution can be understood as a human-readable
presentation of substitution defined on infinitary lambda terms based
on de Bruijn indices.

Infinitary lambda terms could be precisely defined as the
$\alpha$-equivalence classes of the terms given in
Definition~\ref{def_infinitary_lambda_terms}, with a coinductively
defined $\alpha$-equivalence relation~$=_\alpha$. Such a definition
involves some technical issues. If the set of variables~$V$ is
countable, then it may be impossible to choose a ``fresh'' variable
$x \notin \FV(t)$ for a term $t \in \Lambda^\infty$, because~$t$ may
contain all variables free. This presents a difficulty when trying to
precisely define substitution. See
also~\cite{KurzPetrisanSeveriVries2012,KurzPetrisanSeveriVries2013}.
There are two ways of resolving this situation:
\begin{enumerate}
\item assume that~$V$ is uncountable,
\item consider only terms with finitely many free variables.
\end{enumerate}
Assuming that a fresh variable may always be chosen, one may precisely
define substitution and use coinductive techniques to prove: if $t
=_\alpha t'$ and $s =_\alpha s'$ then $s[t/x] =_\alpha s'[t'/x]$. This
implies that substitution lifts to a function on the
$\alpha$-equivalence classes, which is also trivially true for
application and abstraction. Therefore, all functions defined by
guarded corecursion using only the operations of substitution,
application and abstraction lift to functions on $\alpha$-equivalence
classes (provided the same defining equation is used for all terms
within the same $\alpha$-equivalence class). This justifies the use of
Barendregt's variable convention~\cite[2.1.13]{Barendregt1984} (under
the assumption that we may always choose a fresh variable).

Since our formalisation is based on de Bruijn indices, we omit
explicit treatment of $\alpha$-equivalence in this paper.

We also mention that another principled and precise way of dealing
with the renaming of bound variables is to define the set of
infinitary lambda terms as the final coalgebra of an appropriate
functor in the category of nominal
sets~\cite{KurzPetrisanSeveriVries2012,KurzPetrisanSeveriVries2013}.

\begin{defi}
  Let $R \subseteq \Lambda^\infty \times \Lambda^\infty$ be a binary
  relation on infinitary lambda terms. The \emph{compatible closure}
  of~$R$, denoted~$\to_R$, is defined inductively by the following
  rules.
  \[
  \begin{array}{cccc}
    \infer{s \to_R t}{\la s, t \ra \in R} &\quad
    \infer{s t \to_R s' t}{s \to_R s'} &\quad
    \infer{s t \to_R s t'}{t \to_R t'} &\quad
    \infer{\lambda x . s \to_R \lambda x . s'}{s \to_R s'}
  \end{array}
  \]
  If $\la t, s\ra \in R$ then~$t$ is an \emph{$R$-redex}. A term $t
  \in \Lambda^\infty$ is in \emph{$R$-normal form} if there is no $s
  \in \Lambda^\infty$ with $t \to_R s$, or equivalently if it contains
  no $R$-redexes. The \emph{parallel closure} of~$R$, usually
  denoted~$\To_R$, is defined coinductively by the following rules.
  \[
  \begin{array}{cccc}
    \infer={s \To_{R} t}{\la s,t\ra \in R} &~
    \infer={a \To_{R} a}{} &~
    \infer={s_1s_2 \To_{R} t_1t_2}{s_1 \To_{R} t_1 &
      s_2 \To_{R} t_2} &~
    \infer={\lambda x . s \To_{R} \lambda x . s'}{s
      \To_{R} s'}
  \end{array}
  \]
  Let ${\to} \subseteq \Lambda^\infty \times
  \Lambda^\infty$. By~$\reduces$ we denote the transitive-reflexive
  closure of~$\to$, and by~$\to^\equiv$ the reflexive closure
  of~$\to$. The \emph{infinitary closure} of~$\to$, denoted~$\infred$,
  is defined coinductively by the following rules.
  \[
  \begin{array}{ccc}
    \infer={s \infred a}{s \reduces a} &\quad
    \infer={s \infred t_1't_2'}{s \reduces t_1t_2 & t_1 \infred t_1' & t_2 \infred t_2'} &\quad
    \infer={s \infred \lambda x . r'}{s \reduces
      \lambda x . r & r \infred r'}
  \end{array}
  \]

  Let
  $R_\beta = \{ \la (\lambda x . s) t, s[t/x] \ra \mid
  t,s\in\Lambda^\infty \}$. The relation~$\to_\beta$ of \emph{one-step
    $\beta$-reduction} is defined as the compatible closure
  of~$R_\beta$. The relation $\reduces_\beta$ of
  \emph{$\beta$-reduction} is the transitive-reflexive closure
  of~$\to_\beta$. The relation $\infred_\beta$ of \emph{infinitary
    $\beta$-reduction} is defined as the infinitary closure
  of~$\to_\beta$. This gives the same coinductive definition of
  infinitary $\beta$-reduction as in~\cite{EndrullisPolonsky2011}.

  The relation~$\to_w$ of \emph{one-step weak head reduction} is
  defined inductively by the following rules.
  \[
  \begin{array}{cc}
    \infer{(\lambda x . s) t \to_w s[t/x]}{} &\quad
    \infer{s t \to_w s' t}{s \to_w s'}
  \end{array}
  \]
  The relations~$\reduces_w$, $\to_w^\equiv$ and~$\infred_w$ are
  defined accordingly. In a term $(\lambda x . s) t t_1 \ldots t_m$
  the subterm $(\lambda x . s) t$ is the \emph{weak head
    redex}. So~$\to_w$ may contract only the weak head redex.
\end{defi}

\begin{defi}
  Let $\bot$ be a distinguished constant. A $\Lambda^\infty$-term~$t$
  is in \emph{root normal form} (rnf) if:
  \begin{itemize}
  \item $t \equiv a$ with $a \not\equiv \bot$, or
  \item $t \equiv \lambda x . t'$, or
  \item $t \equiv t_1 t_2$ and there is no~$s$ with $t_1 \infred_\beta
    \lambda x . s$ (equivalently, there is no~$s$ with $t_1
    \reduces_\beta \lambda x . s$).
  \end{itemize}
  In other words, a term~$t$ is in rnf if $t \not\equiv \bot$ and~$t$
  does not infinitarily $\beta$-reduce to a $\beta$-redex. We say
  that~$t$ \emph{has rnf} if $t \infred_\beta t'$ for some~$t'$ in
  rnf. In particular, $\bot$ has no rnf. A term with no rnf is also
  called \emph{root-active}. By~$\Rc$ we denote the set of all
  root-active terms.
\end{defi}

\begin{defi}
  A set $\Uc \subseteq \Lambda^\infty$ is a set of \emph{meaningless
    terms} (see~\cite{KennawayVries2003}) if it satisfies the
  following axioms.
  \begin{itemize}
  \item {\bf Closure:} if $t \in \Uc$ and $t \infred_\beta s$ then
    $s \in \Uc$.
  \item {\bf Substitution:} if $t \in \Uc$ then $t[s/x] \in \Uc$ for
    any term~$s$.
  \item {\bf Overlap:} if $\lambda x . s \in \Uc$ then $(\lambda x
    . s) t \in \Uc$.
  \item {\bf Root-activeness:} $\Rc \subseteq \Uc$.
  \item {\bf Indiscernibility:} if $t \in \Uc$ and $t \sim_\Uc s$
    then $s \in \Uc$, where~$\sim_\Uc$ is the parallel closure
    of~$\Uc\times\Uc$.
  \end{itemize}
  A set~$\Uc$ of meaningless terms is a set of \emph{strongly
    meaningless terms} if it additionally satisfies the following
  expansion axiom.
  \begin{itemize}
  \item {\bf Expansion:} if $t \in \Uc$ and $s \infred_\beta t$ then
    $s \in \Uc$.
  \end{itemize}

  Let $\Uc \subseteq \Lambda^\infty$. Let $R_{\bot_\Uc} = \{ \la t, \bot \ra
  \mid t \in \Uc \text{ and } t \not\equiv \bot \}$. We define the
  relation~$\to_{\beta\bot_\Uc}$ of \emph{one-step $\beta\bot_\Uc$-reduction} as
  the compatible closure of $R_{\beta\bot_\Uc} = R_\beta \cup R_{\bot_\Uc}$. A
  term~$t$ is in \emph{$\beta\bot_\Uc$-normal form} if it is in
  $R_{\beta\bot_\Uc}$-normal form. The relation~$\reduces_{\beta\bot_\Uc}$ of
  \emph{$\beta\bot_\Uc$-reduction} is the transitive-reflexive closure
  of~$\to_{\beta\bot_\Uc}$. The relation $\infred_{\beta\bot_\Uc}$ of
  \emph{infinitary $\beta\bot_\Uc$-reduction}, or \emph{B{\"o}hm reduction}
  (over~$\Uc$), is the infinitary closure of~$\to_{\beta\bot_\Uc}$. The
  relation $\To_{\bot_\Uc}$ of \emph{parallel $\bot_\Uc$-reduction} is the
  parallel closure of~$R_{\bot_\Uc}$.
\end{defi}

In general, relations on infinitary terms need to be defined
coinductively. However, if the relation depends only on finite initial
parts of the terms then it may often be defined inductively. Because
induction is generally better understood than coinduction, we prefer
to give inductive definitions whenever it is possible to give such a
definition in a natural way, like with the definition of compatible
closure or one-step weak head reduction. This is in contrast to
e.g.~the definition of infinitary reduction~$\infred$, which
intuitively may contain infinitely many reduction steps, and thus must
be defined by coinduction.

The idea with the definition of the infinitary closure~$\infred$ of a
one-step reduction relation~$\to$ is that the depth at which a redex
is contracted should tend to infinity. This is achieved by
defining~$\infred$ in such a way that always after finitely many
reduction steps the subsequent contractions may be performed only
under a constructor. So the depth of the contracted redex always
ultimately increases. The idea for the definition of~$\infred$ comes
from~\cite{EndrullisPolonsky2011,EndrullisHansenHendriksPolonskySilva2015,EndrullisHansenHendriksPolonskySilva2018}. For
infinitary $\beta$-reduction~$\infred_\beta$ the definition is the
same as in~\cite{EndrullisPolonsky2011}. To each derivation of $t
\infred s$ corresponds a strongly convergent reduction sequence of
length at most~$\omega$ obtained by concatenating the finite
$\reduces$-reductions in the prefixes. See the proof of
Theorem~\ref{thm_strongly_convergent}.

Our definition of meaningless terms differs
from~\cite{KennawayVries2003} in that it treats terms with the~$\bot$
constant, but it is equivalent to the original definition, in the
following sense. Let~$\Lambda_0^\infty$ be the set of
infinitary-lambda terms without~$\bot$. If~$\Uc$ is a set of
meaningless terms defined as in~\cite{KennawayVries2003}
on~$\Lambda_0^\infty$, then~$\Uc_\bot$ (the set of terms from~$\Uc$
with some subterms in~$\Uc$ replaced by~$\bot$) is a set of
meaningless terms according to our definition. Conversely, if~$\Uc$ is
a set of meaningless terms according to our definition, then $\Uc =
\Uc_\bot'$ where $\Uc' = \Uc \cap \Lambda_0^\infty$ ($\Uc'$ then
satisfies the axioms of~\cite{KennawayVries2003}).

To show confluence of B{\"o}hm reduction over~$\Uc$ we also need the
expansion axiom. The reason is purely technical. In the present
coinductive framework there is no way of talking about infinitary
reductions of arbitrary ordinal length, only about reductions of
length~$\omega$. We need the expansion axiom to show that $t
\infred_{\beta\bot_\Uc} t' \To_{\bot_\Uc} s$ implies $t
\infred_{\beta\bot_\Uc} s$.

The expansion axiom is necessary for this implication. Let~$\Os$ be
the \emph{ogre}~\cite{SeveriVries2005} satisfying $\Os \equiv \lambda
x . \Os$, i.e., $\Os \equiv \lambda x_1 . \lambda x_2 . \lambda x_3
\ldots$. A term~$t$ is \emph{head-active}~\cite{SeveriVries2005} if $t
\equiv \lambda x_1 \ldots x_n . r t_1 \ldots t_m$ with $r \in \Rc$ and
$n,m\ge 0$. Define $\Hc = \{ t \in \Lambda^\infty \mid t
\reduces_\beta t' \text{ with } t' \text{ head-active} \}$, $\Oc = \{
t \in \Lambda^\infty \mid t \reduces_\beta \Os \}$ and $\Uc = \Hc \cup
\Oc$. One can show that~$\Uc$ is a set of meaningless terms (see the
appendix). Consider $\Omega_\Os = (\lambda x y . x x) (\lambda x y . x
x)$. We have $\Omega_\Os \infred_\beta \Os \in \Oc$. But $\Omega_\Os
\notin \Uc$, so~$\Uc$ does not satisfy the expansion axiom. Now,
$\Omega_\Os \infred_{\beta\bot_\Uc} \Os \To_{\bot_\Uc} \bot$, but
$\Omega_\Os \not\infred_{\beta\bot_\Uc} \bot$, because no finite
$\beta$-reduct of~$\Omega_\Os$ is in~$\Uc$.

The expansion axiom could probably be weakened slightly, but the
present formulation is simple and it already appeared in the
literature~\cite{SeveriVries2011a,SeveriVries2005,KennawaySeveriSleepVries2005}. Sets
of meaningless terms which do not satisfy the expansion axiom tend to
be artificial. A notion of a set of strongly meaningless terms
equivalent to ours appears in~\cite{SeveriVries2005}. In the presence
of the expansion axiom, the indiscernibility axiom may be
weakened~\cite{SeveriVries2011a,SeveriVries2005}.

In the setup
of~\cite{EndrullisHansenHendriksPolonskySilva2015,EndrullisHansenHendriksPolonskySilva2018}
it is possible to talk about reductions of arbitrary ordinal length,
but we have not investigated the possibility of adapting the framework
of~\cite{EndrullisHansenHendriksPolonskySilva2015,EndrullisHansenHendriksPolonskySilva2018}
to the needs of the present paper.

The axioms of a set~$\Uc$ of meaningless terms are sufficient for
confluence and normalisation of B{\"o}hm reduction
over~$\Uc$. However, they are not necessary. The
paper~\cite{SeveriVries2011} gives axioms necessary and sufficient for
confluence and normalisation.

The following two simple lemmas will often be used implicitly.

\begin{lem}
  Let $\infred$ be the infinitary and~$\reduces$ the
  transitive-reflexive closure of~$\to$. Then the following conditions
  hold for all $t,s,s' \in \Lambda^\infty$:
  \begin{enumerate}
  \item $t \infred t$,
  \item if $t \reduces s \infred s'$ then $t \infred s'$,
  \item if $t \reduces s$ then $t \infred s$.
  \end{enumerate}
\end{lem}

\begin{proof}
  The first point follows by coinduction. The second point follows by
  case analysis on $s \infred s'$. The last point follows from the
  previous two.

  The proof of the first point is straightforward, but to illustrate
  the coinductive technique we give this proof in detail. A reader not
  familiar with coinduction is invited to study this proof and insert
  the implicit ordinals as in Section~\ref{sec_coind}.

  Let $t \in \Lambda^\infty$. There are three cases. If $t \equiv a$
  then $a \reduces a$, so $t \infred t$ by the definition
  of~$\infred$. If $t \equiv t_1t_2$ then $t_1 \infred t_1$ and $t_2
  \infred t_2$ by the coinductive hypothesis. Since also $t \reduces
  t_1t_2$, we conclude $t \infred t$. If $t \equiv \lambda x . t'$
  then $t' \infred t'$ by the coinductive hypothesis. Since also $t
  \reduces \lambda x . t'$, we conclude $t \infred t$.
\end{proof}

\begin{lem}
  If $R \subseteq S \subseteq \Lambda^\infty \times \Lambda^\infty$
  then ${\infred_R} \subseteq {\infred_S}$.
\end{lem}

\begin{proof}
  By coinduction.
\end{proof}

The next three lemmas have essentially been shown
in~\cite[Lemma~4.3--4.5]{EndrullisPolonsky2011}.

\begin{lem}\label{lem_beta_subst}
  If $s \infred_\beta s'$ and $t \infred_\beta t'$ then
  $s[t/x] \infred_\beta s'[t'/x]$.
\end{lem}

\begin{proof}
  By coinduction, with case analysis on $s \infred_\beta s'$, using
  that $t_1 \reduces_\beta t_2$ implies $t_1[t/x] \reduces_\beta
  t_2[t/x]$.
\end{proof}

\begin{lem}\label{lem_beta_beta_fin_append}
  If $t_1 \infred_\beta t_2 \to_\beta t_3$ then $t_1
  \infred_\beta t_3$.
\end{lem}

\begin{proof}
  Induction on $t_2 \to_\beta t_3$, using Lemma~\ref{lem_beta_subst}.
\end{proof}

\begin{lem}\label{lem_beta_append}
  If $t_1 \infred_\beta t_2 \infred_\beta t_3$ then $t_1
  \infred_\beta t_3$.
\end{lem}

\begin{proof}
  By coinduction, with case analysis on $t_2 \infred_\beta t_3$, using
  Lemma~\ref{lem_beta_beta_fin_append}.
\end{proof}

\begin{lem}\label{lem_rnf_fwd}
  If $t$ is in rnf and $t \infred_\beta s$ then $s$ is in rnf.
\end{lem}

\begin{proof}
  Suppose~$s$ is not in rnf, i.e., $s \equiv \bot$ or $s \equiv
  s_1s_2$ with $s_1 \reduces_\beta \lambda x . u$. If $s \equiv \bot$
  then $t \reduces_\beta \bot$, and thus either $t \equiv \bot$ or it
  $\beta$-reduces to a redex. So~$t$ is not in rnf. If $s \equiv
  s_1s_2$ with $s_1 \infred_\beta \lambda x . u'$, then $t
  \reduces_\beta t_1t_2$ with $t_i \infred_\beta s_i$. By
  Lemma~\ref{lem_beta_beta_fin_append} we have $t_1 \infred_\beta
  \lambda x . u$. Thus~$t$ reduces to a redex $(\lambda x . u)
  t_2$. Hence~$t$ is not in rnf.
\end{proof}

\section{Confluence and normalisation of B{\"o}hm
  reductions}\label{sec_inflam_confluence}

In this section we use coinductive techniques to prove confluence and
normalisation of B{\"o}hm reduction over an arbitrary set of strongly
meaningless terms~$\Uc$. The infinitary lambda calculus we are
concerned with, including the $\bot_\Uc$-reductions to~$\bot$, shall
be called the $\lambda_{\beta\bot_\Uc}^\infty$-calculus.

More precisely, our aim is to prove the following theorems.

\medskip

{ \renewcommand{\thethm}{\ref{thm_bohm_cr}}
\begin{thm}[Confluence of the $\lambda_{\beta\bot_\Uc}^\infty$-calculus]~

  If $t \infred_{\beta\bot_\Uc} t_1$ and $t \infred_{\beta\bot_\Uc} t_2$ then
  there exists~$t_3$ such that $t_1 \infred_{\beta\bot_\Uc} t_3$ and $t_2
  \infred_{\beta\bot_\Uc} t_3$.
\end{thm}
\addtocounter{thm}{-1}
}

{ \renewcommand{\thethm}{\ref{thm_bohm_norm}}
\begin{thm}[Normalisation of the $\lambda_{\beta\bot_\Uc}^\infty$-calculus]~

  For every $t \in \Lambda^\infty$ there exists a unique $s \in
  \Lambda^\infty$ in $\beta\bot_\Uc$-normal form such that $t
  \infred_{\beta\bot_\Uc} s$.
\end{thm}
\addtocounter{thm}{-1}
}

In what follows we assume that~$\Uc$ is an arbitrary fixed set of
strongly meaningless terms, unless specified otherwise. Actually,
almost all lemmas are valid for~$\Uc$ being a set of meaningless
terms, without the expansion axiom. Unless explicitly mentioned before
the statement of a lemma, the proofs do not use the expansion
axiom. To show confluence modulo~$\sim_\Uc$
(Theorem~\ref{thm_cr_modulo}), it suffices that~$\Uc$ is a set of
meaningless terms. Confluence and normalisation of the
$\lambda_{\beta\bot_\Uc}$-calculus (Theorem~\ref{thm_bohm_cr} and
Theorem~\ref{thm_bohm_norm}), however, require the expansion
axiom. But this is only because in the present coinductive framework
we are not able to talk about infinite reductions of arbitrary ordinal
length. Essentially, we need the expansion axiom to compress the
B{\"o}hm reductions to length~$\omega$.

The idea of the proof is to show that for every term there exists a
certain standard infinitary $\beta\bot_\Uc$-reduction to normal
form. This reduction is called an infinitary $N_\Uc$-reduction
(Definition~\ref{def_leadsto} and Lemma~\ref{lem_N_normalising}). We
show that the normal forms obtained through infinitary
$N_\Uc$-reductions are unique (Lemma~\ref{lem_N_confluent}). Then we
prove that prepending infinitary $\beta\bot_\Uc$-reduction to an
$N_\Uc$-reduction results in an $N_\Uc$-reduction
(Theorem~\ref{thm_N_prepend}). Since an $N_\Uc$-reduction is an
infinitary $\beta\bot_\Uc$-reduction of a special form
(Lemma~\ref{lem_N_to_bohm}), these results immediately imply
confluence (Theorem~\ref{thm_bohm_cr}) and normalisation
(Theorem~\ref{thm_bohm_norm}) of infinitary
$\beta\bot_\Uc$-reduction. Hence, in essence we derive confluence from
a strengthened normalisation result. See Figure~\ref{fig_cr}.

\begin{figure}[ht]
  \centerline{
    \xymatrix{
      {t_1} \ar@{~>}[d]_>>{N_\Uc} & {t} \ar@{->}[l]_>>{\infty}^>>{\beta\bot_\Uc}
      \ar@{~>}[dl]^>>{N_\Uc} \ar@{->}[r]^>>{\infty}_>>{\beta\bot_\Uc}
      \ar@{~>}[dr]_>>{N_\Uc} & {t_2} \ar@{~>}[d]^>>{N_\Uc} \\
      {t_1'} \ar@{}[rr]|{{\displaystyle\equiv}} & & {t_2'}
    }
  }
  \caption{Confluence of infinitary B{\"o}hm reduction.}\label{fig_cr}
\end{figure}
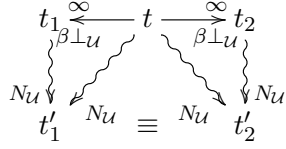

In our proof we use a standardisation result for infinitary
$\beta$-reductions from~\cite{EndrullisPolonsky2011}
(Theorem~\ref{thm_polonsky}). In particular, this theorem is needed to
show uniqueness of canonical root normal forms
(Definition~\ref{def_crnf}). Theorem~\ref{thm_N_prepend} depends on
this. Even when counting in the results
of~\cite{EndrullisPolonsky2011} only referenced here, our confluence
proof may be considered simpler than previous proofs of related
results. In particular, it is much easier for formalise.

We also show that the set of root-active terms is strongly
meaningless. Together with the previous theorems this implies
confluence and normalisation of the
$\lambda_{\beta\bot_\Rc}^\infty$-calculus. Confluence of the
$\lambda_{\beta\bot_\Rc}^\infty$-calculus in turn implies confluence
of~$\infred_\beta$ modulo equivalence of meaningless terms. The
following theorem does not require the expansion axiom.

\medskip

{ \renewcommand{\thethm}{\ref{thm_cr_modulo}}
\begin{thm}[Confluence modulo equivalence of meaningless terms]~

  If $t \sim_\Uc t'$, $t \infred_{\beta} s$ and $t' \infred_{\beta}
  s'$ then there exist~$r,r'$ such that $r \sim_\Uc r'$, $s
  \infred_{\beta} r$ and $s' \infred_{\beta} r'$.
\end{thm}
\addtocounter{thm}{-1}
}

Note that our overall proof strategy is different
from~\cite{Czajka2014,KennawayVries2003,KennawayKlopSleepVries1997}. We
first derive a strengthened normalisation result for B{\"o}hm
reduction, from this we derive confluence of B{\"ohm} reduction, then
we show that root-active terms are strongly meaningless thus
specialising the confluence result, and only using that we show
confluence modulo equivalence of meaningless
terms. In~\cite{Czajka2014,KennawayVries2003,KennawayKlopSleepVries1997}
first confluence modulo equivalence of meaningless terms is shown, and
from that confluence of B{\"o}hm reduction is derived. Of course, some
intermediate lemmas we prove have analogons
in~\cite{Czajka2014,KennawayVries2003,KennawayKlopSleepVries1997}, but
we believe the general proof strategy to be fundamentally different.

\subsection{Properties of~$\sim_\Uc$}

In this subsection~$\Uc$ is an arbitrary fixed set of meaningless
terms, and~$\sim_\Uc$ is the parallel closure of~$\Uc \times \Uc$. The
expansion axiom is not used in this subsection.

\begin{lem}\label{lem_sim_subst_2}
  If $t \sim_\Uc t'$ and $s \sim_\Uc s'$ then $t[s/x] \sim_\Uc
  t'[s'/x]$.
\end{lem}

\begin{proof}
  By coinduction, using the substitution axiom.
\end{proof}

\begin{lem}\label{lem_sim_fin_beta_2}
  If $t \to_{\beta} s$ and $t \sim_{\Uc} t'$ then there is~$s'$ with
  $t' \to_{\beta}^\equiv s'$ and $s \sim_{\Uc} s'$.
\end{lem}

\begin{proof}
  Induction on $t \to_\beta s$. If the case $t,t' \in \Uc$ in the
  definition of $t \sim_\Uc t'$ holds then $s \in \Uc$ by the closure
  axiom, so $t' \sim_\Uc s$ and we may take $s' \equiv t'$. Thus
  assume otherwise. Then all cases follow directly from the inductive
  hypothesis, except when~$t$ is the contracted $\beta$-redex. Then $t
  \equiv (\lambda x . t_1) t_2$ and $s \equiv t_1[t_2/x]$. First
  assume $t \in \Uc$. Then also $t' \in \Uc$ by the indiscernibility
  axiom (note this does not imply that the first case in the
  definition of $t \sim_\Uc t'$ holds). Also $s \in \Uc$ by the
  closure axiom, so $t' \sim_\Uc s$ and we may take $s' \equiv t'$. So
  assume $t \notin \Uc$. Then $\lambda x . t_1 \notin \Uc$ by the
  overlap axiom. Hence $t' \equiv (\lambda x . t_1') t_2'$ with $t_i
  \sim_\Uc t_i'$. Thus $t_1[t_2/x] \sim_\Uc t_1'[t_2'/x]$ by
  Lemma~\ref{lem_sim_subst_2}. So we may take $s' \equiv
  t_1'[t_2'/x]$.
\end{proof}

\begin{lem}\label{lem_sim_beta_2}
  If $t \infred_{\beta} s$ and $t \sim_{\Uc} t'$ then there is~$s'$
  with $t' \infred_{\beta} s'$ and $s \sim_{\Uc} s'$.
\end{lem}

\begin{proof}
  By coinduction. If $s \equiv a$ then $t \reduces_\beta s$ and the
  claim follows from Lemma~\ref{lem_sim_fin_beta_2}. If $s \equiv
  s_1s_2$ then $t \reduces_\beta t_1t_2$ with $t_i \infred_\beta
  s_i$. By Lemma~\ref{lem_sim_fin_beta_2} there is~$u$ with $t_1t_2
  \sim_\Uc u$ and $t' \reduces_\beta u$. If $t_1t_2, u \in \Uc$ then
  $s \in \Uc$ by the closure axiom, and thus we may take $s' \equiv
  u$. Otherwise $u \equiv u_1u_2$ with $t_i \sim_\Uc u_i$. By the
  coinductive hypothesis we obtain $s_1',s_2'$ with $u_i \infred_\beta
  s_i'$ and $s_i \sim_\Uc s_i'$. Take $s' \equiv s_1's_2'$. Then $t'
  \infred_\beta s'$ and $s \sim_\Uc s'$. If $s \equiv \lambda x . s'$
  then the argument is analogous to the previous case.
\end{proof}

\begin{lem}\label{lem_sim_trans}
  If $t \sim_\Uc s$ and $s \sim_\Uc u$ then $t \sim_\Uc s$.
\end{lem}

\begin{proof}
  By coinduction, using the indiscernibility axiom.
\end{proof}

\begin{lem}\label{lem_sim_to_bot}
  If $t \sim_\Uc s$ then there is~$r$ with $t \To_{\bot_\Uc} r$ and $s
  \To_{\bot_\Uc} r$.
\end{lem}

\begin{proof}
  By coinduction.
\end{proof}

\subsection{Properties of parallel $\bot_\Uc$-reduction}

Recall that~$\Uc$ is an arbitrary fixed set of strongly meaningless
terms. The expansion axiom is not used in this subsection except for
Corollary~\ref{cor_back_active}, Lemma~\ref{lem_omega_merge},
Corollary~\ref{cor_bohm_beta_append} and
Lemma~\ref{lem_par_bot_preserves_rnf_rev}.

\begin{lem}\label{lem_subst_omega_1}
  If $s \To_{\bot_\Uc} s'$ and $t \To_{\bot_\Uc} t'$ then $s[t/x]
  \To_{\bot_\Uc} s'[t'/x]$.
\end{lem}

\begin{proof}
  Coinduction with case analysis on $s \To_{\bot_\Uc} s'$, using the
  substitution axiom.
\end{proof}

\begin{lem}\label{lem_omega_to_bohm}
  If $t \To_{\bot_\Uc} s$ then $t \infred_{\beta\bot_\Uc} s$.
\end{lem}

\begin{proof}
  By coinduction.
\end{proof}

\begin{lem}\label{lem_omega_indisc}
  If $t \in \Uc$ and $t \To_{\bot_\Uc} s$ or $s \To_{\bot_\Uc} t$ then
  $s \in \Uc$.
\end{lem}

\begin{proof}
  Using the root-activeness axiom and that $\bot$ is root-active, show
  by coinduction that $t \sim_\Uc s$. Then use the indiscernibility
  axiom.
\end{proof}

\begin{lem}\label{lem_omega_1_collapse}
  If $t_1 \To_{\bot_\Uc} t_2 \To_{\bot_\Uc} t_3$ then $t_1 \To_{\bot_\Uc}
  t_3$.
\end{lem}

\begin{proof}
  Coinduction with case analysis on $t_2 \To_{\bot_\Uc} t_3$, using
  Lemma~\ref{lem_omega_indisc}.
\end{proof}

\begin{lem}\label{lem_omega_1_postpone}
  If $t_1 \To_{\bot_\Uc} t_2 \to_\beta t_3$ then there exists~$t_1'$
  such that $t_1 \to_\beta t_1' \To_{\bot_\Uc} t_3$.
\end{lem}

\begin{proof}
  Induction on $t_2 \to_\beta t_3$. The only interesting case is when
  $t_2 \equiv (\lambda x . s_1) s_2$ and $t_3 \equiv s_1[s_2/x]$. Then
  $t_1 \equiv (\lambda x . u_1) u_2$ with $u_i \To_{\bot_\Uc} s_i$. By
  Lemma~\ref{lem_subst_omega_1}, $u_1[u_2/x] \To_{\bot_\Uc}
  s_1[s_2/x]$. Thus take $t_1' \equiv u_1[u_2/x]$.
\end{proof}

\begin{lem}\label{lem_fin_bohm_decompose}
  If $s \reduces_{\beta\bot_\Uc} t$ then there exists~$r$ such that $s
  \reduces_\beta r \To_{\bot_\Uc} t$.
\end{lem}

\begin{proof}
  Induction on the length of $s \reduces_{\beta\bot_\Uc} t$, using
  Lemma~\ref{lem_omega_1_postpone} and Lemma~\ref{lem_omega_1_collapse}.
\end{proof}

\begin{cor}\label{cor_omega_fin_postpone}
  If $t_1 \To_{\bot_\Uc} t_2 \reduces_{\beta\bot_\Uc} t_3$ then there
  is~$s$ with $t_1 \reduces_\beta s \To_{\bot_\Uc} t_3$.
\end{cor}

\begin{proof}
  Follows from
  Lemmas~\ref{lem_fin_bohm_decompose},~\ref{lem_omega_1_postpone},~\ref{lem_omega_1_collapse}.
\end{proof}

\begin{lem}\label{lem_omega_1_beta_to_beta}
  If $t_1 \To_{\bot_\Uc} t_2 \infred_{\beta\bot_\Uc} t_3$ then $t_1
  \infred_{\beta\bot_\Uc} t_3$.
\end{lem}

\begin{proof}
  By coinduction. There are three cases.
  \begin{itemize}
  \item $t_3 \equiv a$. Then $t_1 \To_{\bot_\Uc} t_2
    \reduces_{\beta\bot_\Uc} a$. By
    Corollary~\ref{cor_omega_fin_postpone} there is~$s$ with $t_1
    \reduces_\beta s \To_{\bot_\Uc} a$. By
    Lemma~\ref{lem_omega_to_bohm} we have $s \infred_{\beta\bot_\Uc}
    a$. Thus $t_1 \infred_{\beta\bot_\Uc} a$.
  \item $t_3 \equiv s_1s_2$. Then $t_1 \To_{\bot_\Uc} t_2
    \reduces_{\beta\bot_\Uc} s_1's_2'$ with $s_i'
    \infred_{\beta\bot_\Uc} s_i$. By
    Corollary~\ref{cor_omega_fin_postpone} there is~$u$ with $t_1
    \reduces_\beta u \To_{\bot_\Uc} s_1's_2'$. Then $u \equiv u_1u_2$
    with $u_i \To_{\bot_\Uc} s_i' \infred_{\beta\bot_\Uc} s_i$. By the
    coinductive hypothesis $u_i \infred_{\beta\bot_\Uc} s_i$. Thus $t_1
    \infred_{\beta\bot_\Uc} s_1s_2 \equiv t_3$.
  \item $t_3 \equiv \lambda x . r$. The argument is analogous to the
    previous case.\qedhere
  \end{itemize}
\end{proof}

The following lemma is an analogon
of~\cite[Lemma~12.9.22]{KennawayVries2003}.

\begin{lem}[Postponement of parallel
  $\bot_\Uc$-reduction]\label{lem_bohm_decompose}~

  If $t \infred_{\beta\bot_\Uc} s$ then there exists~$r$ such that
  $t \infred_\beta r \To_{\bot_\Uc} s$.
\end{lem}

\begin{proof}
  By coinduction with case analysis on $t \infred_{\beta\bot_\Uc} s$,
  using
  Lemmas~\ref{lem_fin_bohm_decompose},~\ref{lem_omega_1_beta_to_beta}.

  Since this is the first of our coinductive proofs involving an
  implicit Skolem function (see Example~\ref{ex_skolem}), we give it
  in detail. The reader is invited to extract from this proof an
  explicit corecursive definition of the Skolem function.

  Assume $t \infred_{\beta\bot_\Uc} s$. There are three cases.
  \begin{itemize}
  \item $s \equiv a$ and $t \reduces_{\beta\bot_\Uc} a$. Then the
    claim follows from Lemma~\ref{lem_fin_bohm_decompose}.
  \item $s \equiv s_1 s_2$ and $t \reduces_{\beta\bot_\Uc} t_1 t_2$
    and $t_i \infred_{\beta\bot_\Uc} s_i$. By
    Lemma~\ref{lem_fin_bohm_decompose} there is $t'$ with $t
    \reduces_\beta t' \To_{\bot_\Uc} t_1t_2$. Because $t_1t_2
    \not\equiv \bot$, we must have $t' \equiv t_1't_2'$ with $t_i'
    \To_{\bot_\Uc} t_i$. By Lemma~\ref{lem_omega_1_beta_to_beta} we
    have $t_i' \infred_{\beta\bot_\Uc} s_i$. By the coinductive
    hypothesis we obtain $s_1',s_2'$ such that $t_i' \infred_\beta
    s_i' \To_{\bot_\Uc} s_i$. Hence $t \infred_\beta s_1's_2'
    \To_{\bot_\Uc} s_1s_2 \equiv s$.
  \item $s \equiv \lambda x . s'$ and $t \reduces_{\beta\bot_\Uc}
    \lambda x . t'$ and $t' \infred_{\beta\bot_\Uc} s'$. By
    Lemma~\ref{lem_fin_bohm_decompose} there is~$u$ with $t
    \reduces_\beta u \To_{\bot_\Uc} \lambda x . t'$. Then $u \equiv
    \lambda x . u'$ with $u' \To_{\bot_\Uc} t'$. By
    Lemma~\ref{lem_omega_1_beta_to_beta} we have $u'
    \infred_{\beta\bot_\Uc} s'$. By the coinductive hypothesis we
    obtain~$w$ such that $u' \infred_\beta w \To_{\bot_\Uc} s'$. Hence
    $t \infred_\beta \lambda x . w \To_{\bot_\Uc} \lambda x . s'
    \equiv s$.\qedhere
  \end{itemize}
\end{proof}

\begin{cor}\label{cor_bohm_preserves_U}
  If $t \in \Uc$ and $t \infred_{\beta\bot_\Uc} s$ then $s \in \Uc$.
\end{cor}

\begin{proof}
  Follows from Lemma~\ref{lem_bohm_decompose}, the closure axiom and
  Lemma~\ref{lem_omega_indisc}.
\end{proof}

The following depend on the expansion axiom.

\begin{cor}\label{cor_back_active}
  If $s \in \Uc$ and $t \infred_{\beta\bot_\Uc} s$ then $t \in \Uc$.
\end{cor}

\begin{proof}
  Follows from Lemma~\ref{lem_bohm_decompose},
  Lemma~\ref{lem_omega_indisc} and the expansion axiom.
\end{proof}

\begin{lem}\label{lem_omega_merge}
  If $t \infred_{\beta\bot_\Uc} t' \To_{\bot_\Uc} s$ then $t
  \infred_{\beta\bot_\Uc} s$.
\end{lem}

\begin{proof}
  By coinduction, analysing $t' \To_{\bot_\Uc} s$. All cases follow
  directly from the coinductive hypothesis, except when $s \equiv
  \bot$ and~$t' \in \Uc$. But then~$t \in \Uc$ by
  Corollary~\ref{cor_back_active}, so $t \To_{\bot_\Uc} s$, and thus $t
  \infred_{\beta\bot_\Uc} s$ by Lemma~\ref{lem_omega_to_bohm}.
\end{proof}

\begin{cor}\label{cor_bohm_beta_append}
  If $t \infred_{\beta\bot_\Uc} s \reduces_{\beta} r$ then $t
  \infred_{\beta\bot_\Uc} r$.
\end{cor}

\begin{proof}
  By Lemma~\ref{lem_bohm_decompose} we have $t \infred_\beta t'
  \To_{\bot_\Uc} s \reduces_{\beta} r$. By
  Lemma~\ref{lem_omega_1_postpone} there is~$s'$ with $t'
  \reduces_\beta s' \To_{\bot_\Uc} r$. By
  Lemma~\ref{lem_beta_beta_fin_append} we have $t \infred_\beta s'$,
  and thus $t \infred_{\beta\bot_\Uc} s'$. By
  Lemma~\ref{lem_omega_merge} we finally obtain $t
  \infred_{\beta\bot_\Uc} r$.
\end{proof}

\begin{lem}\label{lem_par_bot_preserves_rnf_rev}
  If $t \notin U$ and $t \To_{\bot_\Uc} s$ and $s$ is in rnf, then~$t$
  is in rnf.
\end{lem}

\begin{proof}
  We consider possible forms of~$s$.
  \begin{itemize}
  \item $s \equiv a$ with $a \not\equiv \bot$. Then~$t \equiv a$ and~$t$ is
    in rnf.
  \item $s \equiv \lambda x . s'$. Then $t \equiv \lambda x . t'$ with
    $t' \To_{\bot_\Uc} s'$, so~$t$ is in rnf.
  \item $s \equiv s_1 s_2$ and there is no~$r$ with
    $s_1 \infred_\beta \lambda x . r$. Then $t \equiv t_1 t_2$ with
    $t_i \To_{\bot_\Uc} s_i$. Then also $t_1 \sim_\Uc s_1$. Suppose
    $t_1 \infred_\beta \lambda x . r$. By Lemma~\ref{lem_sim_beta_2}
    there is~$r'$ with $s_1 \infred_\beta r' \sim_\Uc \lambda x
    . r$. There are two cases.
    \begin{itemize}
    \item $r', \lambda x . r \in \Uc$. Then
      $(\lambda x . r) t_2 \in \Uc$ by the overlap axiom, and thus
      $t \in \Uc$ by the expansion axiom. Contradiction.
    \item $r' \equiv \lambda x . r''$ with $r \sim_\Uc r''$. But then
      $s_1 \infred_\beta \lambda x . r''$. Contradiction.\qedhere
    \end{itemize}
  \end{itemize}
\end{proof}

\subsection{Weak head reduction}

\begin{thm}[Endrullis,
  Polonsky~\cite{EndrullisPolonsky2011}]\label{thm_polonsky}~
  $t \infred_{\beta} s$ iff $t \infred_{w} s$.
\end{thm}

Strictly speaking, in~\cite{EndrullisPolonsky2011} the above theorem
is shown for a different set of infinitary lambda terms which do not
contain constants. However, it is clear that for the purposes
of~\cite{EndrullisPolonsky2011} constants may be treated as variables
not occuring bound, and thus the proof of the above theorem may be
used in our setting. We omit the proof of this theorem here, but we
included the proof in our formalisation.

\begin{lem}\label{lem_w_fin_commute}
  If $t \infred_w t_1$ and $t \to_w t_2$ then there is~$t_3$ with $t_2
  \infred_w t_3$ and $t_1 \to_w^\equiv t_3$.
\end{lem}

\begin{proof}
  If the weak head redex in~$t$ is contracted in $t \infred_w t_1$
  then $t \to_w t_2 \infred_w t_1$ and we may take $t_3 \equiv
  t_1$. Otherwise $t \equiv (\lambda x . s) u u_1 \ldots u_m$, $t_2
  \equiv s[u/x] u_1 \ldots u_m$ and $t_1 \equiv (\lambda x . s') u'
  u_1' \ldots u_m'$ with $s \infred_w s'$, $u \infred_w u'$ and $u_i
  \infred_w u_i'$ for $i=1,\ldots,m$. By Theorem~\ref{thm_polonsky}
  and Lemma~\ref{lem_beta_subst} we obtain $s[u/x] \infred_w
  s'[u'/x]$. Take $t_3 \equiv s'[u'/x] u_1' \ldots u_m'$. Then $t_2
  \infred_w t_3$ and $t_1 \to_w t_3$.
\end{proof}

\begin{lem}\label{lem_beta_to_rnf}
  If $t \infred_\beta s$ with~$s$ in rnf, then there is~$s'$ in rnf
  with $t \reduces_w s' \infred_w s$.
\end{lem}

\begin{proof}
  By Theorem~\ref{thm_polonsky} we have $t \infred_w s$. Because~$s$
  is in rnf, there are three cases.
  \begin{itemize}
  \item $s \equiv a$ with $a \not\equiv \bot$. Then $t \reduces_w s$
    and we may take $s' \equiv s$.
  \item $s \equiv \lambda x . s_0$. Then $t \reduces_w \lambda x
    . t_0$ with $t_0 \infred_w s_0$. So take $s' \equiv \lambda x
    . t_0$.
  \item $s \equiv s_1 s_2$ and there is no~$r$ with $s_1 \infred_\beta
    \lambda x . r$. Then $t \reduces_w t_1 t_2$ with $t_i \infred_w
    s_i$. Suppose $t_1 \infred_\beta \lambda x . u$. Then $t_1
    \reduces_w \lambda x . u'$ for some~$u'$, by
    Theorem~\ref{thm_polonsky}. By Lemma~\ref{lem_w_fin_commute} there
    is~$r$ with $\lambda x . u' \infred_w r$ and $s_1 \reduces_w
    r$. But then $r \equiv \lambda x . r'$, so~$s_1$ reduces to an
    abstraction. Contradiction. Hence~$t_1t_2$ is in rnf, so we may
    take $s' \equiv t_1t_2$.\qedhere
  \end{itemize}
\end{proof}

\begin{lem}\label{lem_w_unique}
  If $t \reduces_w s_1$, $t \reduces_w s_2$ and these reductions have
  the same length, then $s_1 \equiv s_2$.
\end{lem}

\begin{proof}
  By induction on the length of the reduction, using the fact that
  weak head redexes are unique if they exist.
\end{proof}

\newcommand{\crnf}{\mathrm{crnf}}

\begin{defi}\label{def_crnf}
  The \emph{canonical root normal form} (crnf) of a term~$t$ is an
  rnf~$s$ such that $t \reduces_w s$ and this reduction is shortest
  among all finitary weak head reductions of~$t$ to root normal form.
\end{defi}

It follows from Lemma~\ref{lem_beta_to_rnf} and
Lemma~\ref{lem_w_unique} that if~$t$ has a rnf then it has a unique
crnf. We shall denote this crnf by~$\crnf(t)$.

\begin{lem}\label{lem_beta_to_crnf}
  If $t \infred_\beta s$ with~$s$ in rnf, then $t \reduces_w \crnf(t)
  \infred_w s$.
\end{lem}

\begin{proof}
  Follows from Lemma~\ref{lem_beta_to_rnf} and
  Lemma~\ref{lem_w_unique}.
\end{proof}

\subsection{Infinitary $N_\Uc$-reduction}\label{sec_inf_n_red}

In the $\lambda_{\beta\bot_\Uc}^\infty$-calculus every term has a
unique normal form. This normal form may be obtained through an
infinitary $N_\Uc$-reduction, defined below.

\begin{defi}\label{def_leadsto}
  The relation~$\leadsto_{N_\Uc}$ is defined coinductively.
  \[
  \begin{array}{c}
    \infer={t \leadsto_{N_\Uc} a}{t \notin \Uc & \crnf(t) \equiv a} \quad\quad\,
    \infer={t \leadsto_{N_\Uc} s_1s_2}{t \notin \Uc & \crnf(t) \equiv t_1t_2 & t_1
      \leadsto_{N_\Uc} s_1 & t_2 \leadsto_{N_\Uc} s_2} \\ \\
    \infer={t \leadsto_{N_\Uc} \lambda x . s}{t \notin \Uc & \crnf(t) \equiv \lambda x . t'
      & t' \leadsto_{N_\Uc} s} \quad\quad\,
    \infer={t \leadsto_{N_\Uc} \bot}{t \in \Uc}
  \end{array}
  \]
\end{defi}

Note that because $\Rc \subseteq \Uc$, every term $t \notin \Uc$ has a
rnf, so~$\crnf(t)$ is defined for $t \notin \Uc$. Also note that
$\leadsto_{N_\Uc}$ is not closed under contexts, e.g.,
$t \leadsto_{N_\Uc} t'$ does \emph{not} imply
$t s \leadsto_{N_\Uc} t' s$.

The infinitary $N_\Uc$-reduction~$\leadsto_{N_\Uc}$ reduces a term to
its normal form --- its B{\"o}hm-like tree. It is a ``standard''
reduction with a specifically regular structure, which allows us to
prove Theorem~\ref{thm_N_prepend}: if
$t \infred_{\beta\bot_\Uc} t' \leadsto_{N_\Uc} s$ then
$t \leadsto_{N_\Uc} s$. This property allows us to derive confluence
from the fact that every term has a unique normal form reachable by an
infinitary $N_\Uc$-reduction. See Figure~\ref{fig_cr}. It is crucial
here that canonical root normal forms are unique, and that
Lemma~\ref{lem_beta_to_crnf} holds. This depends on
Theorem~\ref{thm_polonsky} --- the standardisation result shown by
Endrullis and Polonsky.

\begin{lem}\label{lem_N_to_bohm}
  If $t \leadsto_{N_\Uc} s$ then $t \infred_{\beta\bot_\Uc} s$.
\end{lem}

\begin{proof}
  By coinduction.
\end{proof}

\begin{lem}\label{lem_N_normalising}
  For every term $t \in \Lambda^\infty$ there is~$s$ with $t
  \leadsto_{N_\Uc} s$.
\end{lem}

\begin{proof}
  By coinduction. If~$t \in \Uc$ then $t \leadsto_{N_\Uc} \bot$ and we
  may take $s \equiv \bot$. Otherwise there are three cases depending
  on the form of~$\crnf(t)$.
  \begin{itemize}
  \item $\crnf(t) \equiv a$. Then $t \leadsto_{N_\Uc} a$ by the first
    rule, so we may take $s \equiv a$.
  \item $\crnf(t) \equiv t_1t_2$. Then by the coinductive hypothesis
    we obtain $s_1,s_2$ with $t_i \leadsto_{N_\Uc} s_i$. Take
    $s \equiv s_1s_2$. Then $t \leadsto_{N_\Uc} s$.
  \item $\crnf(t) \equiv \lambda x . t'$. Analogous to the previous
    case.\qedhere
  \end{itemize}
\end{proof}

\begin{lem}\label{lem_N_confluent}
  If $t \leadsto_{N_\Uc} s_1$ and $t \leadsto_{N_\Uc} s_2$ then
  $s_1 \equiv s_2$.
\end{lem}

\begin{proof}
  By coinduction. If $s_1 \equiv \bot$ then~$t \in \Uc$, so we must
  also have $s_2 \equiv \bot$. Otherwise there are three cases,
  depending on the form of~$\crnf(t)$. Suppose
  $\crnf(t) \equiv t_1t_2$, other cases being similar. Then
  $s_1 \equiv u_1u_2$ with $t_i \leadsto_{N_\Uc} u_i$ and
  $s_2 \equiv w_1w_2$ with $t_i \leadsto_{N_\Uc} w_i$. By the
  coinductive hypothesis $u_i \equiv w_i$. Thus
  $s_1 \equiv u_1u_2 \equiv w_1w_2 \equiv s_2$.
\end{proof}

The next two lemmas and the theorem depend on the expansion axiom.

\begin{lem}\label{lem_N_nf}
  If $t \leadsto_{N_\Uc} s$ then $s$ is in $\beta\bot_\Uc$-normal
  form.
\end{lem}

\begin{proof}
  Suppose~$s$ contains a $\beta\bot_\Uc$-redex. Without loss of
  generality, assume the redex is at the root. First assume that~$s$
  is a $\bot_\Uc$-redex, i.e., $s \in \Uc$ and $s \not\equiv
  \bot$. Using Lemma~\ref{lem_N_to_bohm} we conclude
  $t \infred_{\beta\bot_\Uc} s$. Then $t \in \Uc$ by
  Corollary~\ref{cor_back_active}. Thus $s \equiv \bot$, giving a
  contradiction. So assume~$s$ is a $\beta$-redex, i.e.,
  $s \equiv (\lambda x . s_1) s_2$. But by inspecting the definition
  of $t \leadsto_{N_\Uc} s$ one sees that this is only possible
  when~$\crnf(t)$ is a $\beta$-redex. Contradiction.
\end{proof}

\begin{lem}\label{lem_omega_rnf}
  Suppose $t \infred_{\beta\bot_\Uc} s$ and $t,s$ are in rnf.
  \begin{itemize}
  \item If $s \equiv a$ then $t \equiv s$.
  \item If $s \equiv s_1s_2$ then $t \equiv t_1t_2$ with $t_i
    \infred_{\beta\bot_\Uc} s_i$.
  \item If $s \equiv \lambda x . s'$ then $t \equiv \lambda x . t'$
    with $t' \infred_{\beta\bot_\Uc} s'$.\qedhere
  \end{itemize}
\end{lem}

\begin{proof}
  First note that by Lemma~\ref{lem_bohm_decompose} there is~$r$
  with $t \infred_\beta r \To_{\bot_\Uc} s$.
  \begin{itemize}
  \item If $s \equiv a$ then $a \not\equiv \bot$ and $r \equiv a$, and thus
    $t \reduces_\beta a$. But because~$t$ is in rnf it does not reduce
    to a $\beta$-redex, so in fact $t \equiv a$.
  \item If $s \equiv s_1s_2$ then $r \equiv r_1r_2$ with
    $r_i \To_{\bot_\Uc} s_i$. Thus $t \reduces_\beta t_1't_2'$ where
    $t_i' \infred_\beta r_i$. Because~$t$ is in rnf, we must in fact
    have $t \equiv t_1t_2$ with $t_i \reduces_\beta t_i'$. Then
    $t_i \infred_\beta r_i \To_{\bot_\Uc} s_i$, so
    $t_i \infred_{\beta\bot_\Uc} s_i$ by Lemma~\ref{lem_omega_merge}.
  \item The case $s \equiv \lambda x . s'$ is analogous to the
    previous one.\qedhere
  \end{itemize}
\end{proof}

\begin{thm}\label{thm_N_prepend}
  If $t \infred_{\beta\bot_\Uc} t' \leadsto_{N_\Uc} s$ then
  $t \leadsto_{N_\Uc} s$.
\end{thm}

\begin{proof}
  By coinduction. If $s \equiv \bot$ then~$t' \in \Uc$. By
  Corollary~\ref{cor_back_active} also~$t \in \Uc$. Hence
  $t \leadsto_{N_\Uc} \bot \equiv s$. If $s \not\equiv\bot$ then
  $t' \notin \Uc$ and $t' \reduces_w \crnf(t')$. By
  Corollary~\ref{cor_bohm_beta_append} we have
  $t \infred_{\beta\bot_\Uc} \crnf(t')$. By
  Lemma~\ref{lem_bohm_decompose} there is~$r$ with
  $t \infred_\beta r \To_{\bot_\Uc} \crnf(t')$. We have $t \notin \Uc$
  by Corollary~\ref{cor_bohm_preserves_U}. Then~$r$ is in rnf by
  Lemma~\ref{lem_par_bot_preserves_rnf_rev}. Hence
  $t \reduces_w \crnf(t) \infred_{\beta\bot_\Uc} \crnf(t')$ by
  Lemma~\ref{lem_beta_to_crnf} and Lemma~\ref{lem_omega_merge}. Now
  there are three cases depending on the form of~$\crnf(t')$.
  \begin{itemize}
  \item $\crnf(t') \equiv a$. Then $s \equiv a$, and $\crnf(t) \equiv
    a$ by Lemma~\ref{lem_omega_rnf}. Thus $t \leadsto_{N_\Uc} a \equiv s$.
  \item $\crnf(t') \equiv t_1't_2'$. Then $s \equiv s_1s_2$ with
    $t_i' \leadsto_{N_\Uc} s_i$. By Lemma~\ref{lem_omega_rnf} we have
    $\crnf(t) \equiv t_1t_2$ with $t_i \infred_{\beta\bot_\Uc}
    t_i'$. By the coinductive hypothesis $t_i \leadsto_{N_\Uc}
    s_i$. Hence $t \leadsto_{N_\Uc} s_1s_2 \equiv s$.
  \item The case $\crnf(t') \equiv \lambda x . u$ is analogous to the
    previous one.\qedhere
  \end{itemize}
\end{proof}

\subsection{Confluence and normalisation}

Recall that~$\Uc$ is an arbitrary fixed set of strongly meaningless
terms.

\begin{thm}[Confluence of the $\lambda_{\beta\bot_\Uc}^\infty$-calculus]\label{thm_bohm_cr}~

  If $t \infred_{\beta\bot_\Uc} t_1$ and
  $t \infred_{\beta\bot_\Uc} t_2$ then there exists~$t_3$ such that
  $t_1 \infred_{\beta\bot_\Uc} t_3$ and
  $t_2 \infred_{\beta\bot_\Uc} t_3$.
\end{thm}

\begin{proof}
  By Lemma~\ref{lem_N_normalising} there are~$t_1',t_2'$ with
  $t_i \leadsto_{N_\Uc} t_i'$ for $i=1,2$. By
  Theorem~\ref{thm_N_prepend} we have $t \leadsto_{N_\Uc} t_i'$ for
  $i=1,2$. By Lemma~\ref{lem_N_confluent} we have $t_1' \equiv
  t_2'$. Take $t_3 \equiv t_1' \equiv t_2'$. We have
  $t_i \leadsto_{N_\Uc} t_3$ for $i=1,2$, so
  $t_1 \infred_{\beta\bot_\Uc} t_3$ and
  $t_2 \infred_{\beta\bot_\Uc} t_3$ by Lemma~\ref{lem_N_to_bohm}.
\end{proof}

\begin{thm}[Normalisation of the
  $\lambda_{\beta\bot_\Uc}^\infty$-calculus]\label{thm_bohm_norm}~

  For every $t \in \Lambda^\infty$ there exists a unique
  $s \in \Lambda^\infty$ in $\beta\bot_\Uc$-normal form such that
  $t \infred_{\beta\bot_\Uc} s$.
\end{thm}

\begin{proof}
  By Lemma~\ref{lem_N_normalising} there is~$s$ with
  $t \leadsto_{N_\Uc} s$. By Lemma~\ref{lem_N_nf}, $s$ is in
  $\beta\bot_\Uc$-normal form. By Lemma~\ref{lem_N_to_bohm} we have
  $t \infred_{\beta\bot_\Uc} s$. The uniqueness of~$s$ follows from
  Theorem~\ref{thm_bohm_cr}.
\end{proof}

\subsection{Root-active terms are strongly meaningless}

\begin{defi}\label{def_succ}
  We define the relation~$\succ_x$ coinductively
  \[
  \infer={t \succ_x xu_1\ldots u_n}{u_1,\ldots,u_n \in \Lambda^\infty}\quad\quad
  \infer={a \succ_x a}{}\quad\quad
  \infer={ts \succ_x t's'}{t \succ_x t' & s \succ_x s'}\quad\quad
  \infer={\lambda y . t \succ_x \lambda y . t'}{t \succ_x t' & x \ne y}
  \]
  In other words, $s \succ_x s'$ iff~$s'$ may be obtained from~$s$ by
  changing some arbitrary subterms in~$s$ into some terms having the
  form $xu_1\ldots u_n$.
\end{defi}

\begin{lem}\label{lem_subst_succ}
  If $t \succ_x t'$, $s \succ_x s'$ and $x \ne y$ then $t[s/y] \succ_x
  t'[s'/y]$.
\end{lem}

\begin{proof}
  By coinduction, analysing $t \succ_x t'$.
\end{proof}

\begin{lem}\label{lem_beta_succ}
  If $t \succ_x s$ and $t \to_\beta t'$ then there is~$s'$ with $t'
  \succ_x s'$ and $s \to_\beta^\equiv s'$.
\end{lem}

\begin{proof}
  Induction on~$t \to_\beta t'$. The interesting case is when $t
  \equiv (\lambda y . t_1) t_2$, $t' \equiv t_1[t_2/y]$, $s \equiv
  s_1s_2$, $\lambda y . t_1 \succ_x s_1$ and $t_2 \succ_x s_2$. If
  $s_1 \equiv x u_1 \ldots u_m$ then $t' \succ_x x u_1 \ldots u_m s_2$
  and we may take $s' \equiv s$. Otherwise $s_1 \equiv \lambda y
  . s_1'$ with $t_1 \succ_x s_1'$ (by the variable convention $x \ne
  y$). Then $t' \equiv t_1[t_2/y] \succ_x s_1'[s_2/y]$ by
  Lemma~\ref{lem_subst_succ}. We may thus take $s' \equiv
  s_1'[s_2/y]$.
\end{proof}

\begin{lem}\label{lem_beta_succ_rev}
  If $t \succ_x s$ and $s \to_\beta s'$ then there is~$t'$ with $t'
  \succ_x s'$ and $t \to_\beta^\equiv t'$.
\end{lem}

\begin{proof}
  Induction on~$s \to_\beta s'$, using Lemma~\ref{lem_subst_succ} for
  the redex case.
\end{proof}

\begin{lem}\label{lem_succ_rnf}
  If $t \succ_x t'$ and~$t$ is in rnf, then so is~$t'$.
\end{lem}

\begin{proof}
  Assume~$t'$ is not in rnf. Then $t' \equiv \bot$ or $t' \equiv
  t_1't_2'$ with~$t_1'$ reducing to an abstraction. If $t' \equiv
  \bot$ then $t \equiv \bot$, so assume $t' \equiv t_1't_2'$ and $t_1'
  \reduces_\beta \lambda y . u'$ with $x \ne y$. Then $t \equiv
  t_1t_2$ with $t_i \succ_x t_i'$. By Lemma~\ref{lem_beta_succ_rev}
  there is~$u$ with $t_1 \reduces_\beta \lambda y . u$ and $u \succ_x
  u'$. But this implies that $t \equiv t_1t_2$ is not in
  rnf. Contradiction.
\end{proof}

\begin{lem}\label{lem_subst_active}
  If $t_1, t_2 \in \Lambda^\infty$ and~$t_1$ has no rnf, then neither
  does $t_1[t_2/x]$.
\end{lem}

\begin{proof}
  Assume $t_1[t_2/x]$ has a rnf. Then $t_1[t_2/x] \reduces_\beta s$
  for some~$s$ in rnf, by Lemma~\ref{lem_beta_to_rnf}. By the variable
  convention $t_1[t_2/x] \succ_x t_1$. Hence by
  Lemma~\ref{lem_beta_succ} there is~$s'$ such that $t_1
  \reduces_\beta s'$ and $s \succ_x s'$. Since~$s$ is in rnf, so
  is~$s'$, by Lemma~\ref{lem_succ_rnf}. Thus~$t_1$ has a
  rnf. Contradiction.
\end{proof}

\begin{lem}\label{lem_sim_subst}
  If $t \sim_\Rc t'$ and $s \sim_\Rc s'$ then $t[s/x] \sim_\Rc
  t'[s'/x]$.
\end{lem}

\begin{proof}
  By coinduction, using Lemma~\ref{lem_subst_active}.
\end{proof}

\begin{lem}\label{lem_beta_sim}
  If $t \to_\beta t'$ and $t \sim_\Rc s$ then there is~$s'$ with $s
  \to_\beta^\equiv s'$ and $t' \sim_\Rc s'$.
\end{lem}

\begin{proof}
  Induction on $t \to_\beta t'$. There are two interesting cases.
  \begin{itemize}
  \item $t, s \in \Rc$, i.e., they have no rnf. Then also $t' \in
    \Rc$, so we may take $s' \equiv s$.
  \item $t \equiv (\lambda x . t_1) t_2$, $t' \equiv t_1[t_2/x]$, $s
    \equiv (\lambda x . s_1) s_2$ and $t_i \sim_\Rc s_i$. Then $t'
    \sim_\Rc s_1[s_2/x]$ by Lemma~\ref{lem_sim_subst}. Hence we may
    take $s' \equiv s_1[s_2/x]$.\qedhere
  \end{itemize}
\end{proof}

\begin{lem}\label{lem_sim_rnf}
  If $t$ is in rnf and $t \sim_\Rc s$, then so is~$s$.
\end{lem}

\begin{proof}
  Because~$t$ is in rnf, there are three cases.
  \begin{itemize}
  \item $t \equiv a$ with $a \not\equiv \bot$. Then $s \equiv t$, so
    it is in rnf.
  \item $t \equiv \lambda x . t'$. Then $s \equiv \lambda x . s'$,
    so~$s$ is in rnf.
  \item $t \equiv t_1t_2$ and~$t_1$ does not $\beta$-reduce to an
    abstraction. Then $s \equiv s_1s_2$ with $t_i \sim_\Rc
    s_i$. Assume $s_1 \reduces_\beta \lambda x . s'$. Then by
    Lemma~\ref{lem_beta_sim} there is~$t'$ with $t_1 \reduces_\beta t'
    \sim_\Rc \lambda x . s'$. But then~$t'$ must be an
    abstraction. Contradiction.\qedhere
  \end{itemize}
\end{proof}

\begin{cor}\label{cor_sim_rnf}
  If $t$ has a rnf and $t \sim_\Rc s$, then so does~$s$.
\end{cor}

\begin{proof}
  Follows from Lemma~\ref{lem_beta_sim} and Lemma~\ref{lem_sim_rnf}.
\end{proof}

\begin{lem}\label{lem_infbeta_rnf}
  If $t \infred_\beta s$ and~$t$ has a rnf, then so does~$s$.
\end{lem}

\begin{proof}
  Suppose~$t$ has a rnf. Then by Lemma~\ref{lem_beta_to_rnf} there
  is~$t'$ in rnf with $t \reduces_w t'$. By Theorem~\ref{thm_polonsky}
  and Lemma~\ref{lem_w_fin_commute} there is~$r$ with $s \reduces_w r$
  and $t' \infred_\beta r$. Since~$t'$ is in rnf, by
  Lemma~\ref{lem_rnf_fwd} so is~$r$. Hence~$s$ has a rnf~$r$.
\end{proof}

\begin{thm}\label{thm_root_active_meaningless}
  The set~$\Rc$ of root-active terms is a set of strongly meaningless
  terms.
\end{thm}

\begin{proof}
  We check the axioms. The root-activeness axiom is obvious. The
  closure axiom follows from Lemma~\ref{lem_beta_append}. The
  substitution axiom follows from Lemma~\ref{lem_subst_active}. The
  overlap axiom follows from the fact that lambda abstractions are in
  rnf. The indiscernibility axiom follows from
  Corollary~\ref{cor_sim_rnf}. The expansion axiom follows from
  Lemma~\ref{lem_infbeta_rnf}.
\end{proof}

\begin{cor}[Confluence of the $\lambda_{\beta\bot_\Rc}^\infty$-calculus]\label{cor_bohm_ra_cr}~

  If $t \infred_{\beta\bot_\Rc} t_1$ and $t \infred_{\beta\bot_\Rc} t_2$ then
  there exists~$t_3$ such that $t_1 \infred_{\beta\bot_\Rc} t_3$ and $t_2
  \infred_{\beta\bot_\Rc} t_3$.
\end{cor}

\begin{cor}[Normalisation of the
  $\lambda_{\beta\bot_\Rc}^\infty$-calculus]\label{cor_bohm_ra_norm}~

  For every $t \in \Lambda^\infty$ there exists a unique $s \in
  \Lambda^\infty$ in $\beta\bot_\Rc$-normal form such that $t
  \infred_{\beta\bot_\Rc} s$.
\end{cor}

\subsection{Confluence modulo equivalence of meaningless terms}

From confluence of the $\lambda_{\beta\bot_\Rc}^\infty$-calculus we
may derive confluence of infinitary $\beta$-reduction~$\infred_\beta$
modulo equivalence of meaningless terms. The expansion axiom in not
needed for the proof of the following theorem.

\begin{thm}[Confluence modulo equivalence of meaningless
  terms]\label{thm_cr_modulo}~

  If $t \sim_\Uc t'$, $t \infred_{\beta} s$ and $t' \infred_{\beta}
  s'$ then there exist~$r,r'$ such that $r \sim_\Uc r'$, $s
  \infred_{\beta} r$ and $s' \infred_{\beta} r'$.
\end{thm}

\begin{proof}
  By Lemma~\ref{lem_sim_beta_2} and Lemma~\ref{lem_sim_trans} it
  suffices to consider the case $t \equiv t'$. By
  Corollary~\ref{cor_bohm_ra_cr} there is~$u$ with $s
  \infred_{\beta\bot_\Rc} u$ and $s' \infred_{\beta\bot_\Rc} u$. By
  Theorem~\ref{thm_root_active_meaningless} and
  Lemma~\ref{lem_bohm_decompose} there are $r,r'$ with $s
  \infred_\beta r \sim_{\Rc} u$ and $s' \infred_\beta r' \sim_\Rc
  u$. Because $\Rc \subseteq \Uc$, by Lemma~\ref{lem_sim_trans} we
  obtain $r \sim_\Uc r'$.
\end{proof}

\section{Strongly convergent reductions}\label{sec_inflam_strong_convergence}

In this section we prove that the existence of coinductive infinitary
reductions is equivalent to the existence of strongly convergent
reductions, under certain assumptions. As a corollary, this also
yields $\omega$-compression of strongly convergent reductions, under
certain assumptions. The equivalence proof is virtually the same as
in~\cite{EndrullisPolonsky2011}. The notion of strongly convergent
reductions is the standard notion of infinitary reductions used in
non-coinductive treatments of infinitary lambda calculus.

\newcommand{\aux}{\ensuremath{\to^{2\infty}}}
\newcommand{\infreds}{\ensuremath{\to^{\infty*}}}

\begin{defi}
  On the set of infinitary lambda terms we define a metric~$d$ by
  \[
  d(t,s) = \inf\{2^{-n} \mid t^{\upharpoonright n} \equiv
  s^{\upharpoonright{n}} \}
  \]
  where $r^{\upharpoonright n}$ for $r \in \Lambda^\infty$ is defined
  as the infinitary lambda term obtained by replacing all subterms
  of~$r$ at depth~$n$ by~$\bot$. This defines a metric topology on the
  set of infinitary lambda terms. Let
  $R \subseteq \Lambda^\infty \times \Lambda^\infty$ and let~$\zeta$
  be an ordinal. A map $f : \{\gamma \le \zeta\} \to \Lambda^\infty$
  together with reduction steps
  $\sigma_\gamma : f(\gamma) \to_R f(\gamma+1)$ for $\gamma < \zeta$ is a
  \emph{strongly convergent $R$-reduction sequence of length~$\zeta$
    from~$f(0)$ to~$f(\zeta)$} if the following conditions hold:
  \begin{enumerate}
  \item if $\delta \le \zeta$ is a limit ordinal then~$f(\delta)$ is
    the limit in the metric topology on infinite terms of the
    ordinal-indexed sequence $(f(\gamma))_{\gamma<\delta}$,
  \item if $\delta \le \zeta$ is a limit ordinal then for every $d
    \in \Nbb$ there exists $\gamma < \delta$ such that for all~$\gamma'$
    with $\gamma \le \gamma' < \delta$ the redex contracted in the
    step~$\sigma_{\gamma'}$ occurs at depth greater than~$d$.
  \end{enumerate}
  We write $s \xto{S,\zeta}_R t$ if~$S$ is a strongly convergent
  $R$-reduction sequence of length~$\zeta$ from~$s$ to~$t$.

  A relation ${\to} \subseteq \Lambda^\infty \times \Lambda^\infty$ is
  \emph{appendable} if $t_1 \infred t_2 \to t_3$ implies $t_1 \infred
  t_3$. We define~$\aux$ as the infinitary closure of~$\infred$. We
  write~$\infreds$ for the transitive-reflexive closure of~$\infred$.
\end{defi}

\begin{lem}\label{lem_appendable_append}
  If $\to$ is appendable then $t_1 \infred t_2 \infred t_3$ implies $t_1
  \infred t_3$.
\end{lem}

\begin{proof}
  By coinduction. This has essentially been shown
  in~\cite[Lemma~4.5]{EndrullisPolonsky2011}.
\end{proof}

\begin{lem}\label{lem_two_to_one}
  If $\to$ is appendable then $s \aux t$ implies $s \infred t$.
\end{lem}

\begin{proof}
  By coinduction. There are three cases.
  \begin{itemize}
  \item $t \equiv a$. Then $s \infreds a$, so $s \infred a$ by
    Lemma~\ref{lem_appendable_append}.
  \item $t \equiv t_1 t_2$. Then there are $t_1',t_2'$ with $s
    \infreds t_1't_2'$ and $t_i' \aux t_i$. By
    Lemma~\ref{lem_appendable_append} we have $s \infred t_1't_2'$, so
    there are $u_1,u_2$ with $s \reduces u_1u_2$ and $u_i \infred
    t_i'$. Then $u_i \aux t_i$. By the coinductive hypothesis $u_i
    \infred t_i$. Hence $s \infred t_1t_2 \equiv t$.
  \item $t \equiv \lambda x . r$. Then by
    Lemma~\ref{lem_appendable_append} there is~$s'$ with $s \infred
    \lambda x . s'$ and $s' \aux r$. So there is~$s_0$ with $s
    \reduces \lambda x . s_0$ and $s_0 \infred s'$. Then also $s_0
    \aux r$. By the coinductive hypothesis $s_0 \infred r$. Thus $s
    \infred \lambda x . r \equiv t$.\qedhere
  \end{itemize}
\end{proof}

\begin{thm}\label{thm_strongly_convergent}
  For every $R \subseteq \Lambda^\infty \times \Lambda^\infty$ such
  that~$\to_R$ is appendable, and for all $s,t \in \Lambda^\infty$, we
  have the equivalence: $s \infred_R t$ iff there exists a strongly
  convergent $R$-reduction sequence from~$s$ to~$t$. Moreover, if $s
  \infred_R t$ then the sequence may be chosen to have length at
  most~$\omega$.
\end{thm}

\begin{proof}
  The proof is a straightforward generalisation of the proof of
  Theorem~3 in~\cite{EndrullisPolonsky2011}.

  Suppose that $s \infred_R t$. By traversing the infinite derivation
  tree of $s \infred_R t$ and accumulating the finite prefixes by
  concatenation, we obtain a reduction sequence of length at
  most~$\omega$ which satisfies the depth requirement by construction.

  For the other direction, by induction on~$\zeta$ we show that if $s
  \xto{S,\zeta}_R t$ then $s \aux_R t$, which suffices for $s
  \infred_R t$ by Lemma~\ref{lem_two_to_one}. There are three cases.
  \begin{itemize}
  \item $\zeta = 0$. If $s \xto{S,0}_R t$ then $s \equiv t$, so $s
    \aux_R t$.
  \item $\zeta = \gamma+1$. If $s \xto{S,\gamma+1}_R t$ then $s
    \xto{S',\gamma}_R s' \to_R t$. Hence $s \aux_R s'$ by the
    inductive hypothesis. Then $s \infred_R s' \to_R t$ by
    Lemma~\ref{lem_two_to_one}. So $s \infred_R t$ because~$\to_R$ is
    appendable.
  \item $\zeta$ is a limit ordinal. By coinduction we show that if $s
    \xto{S,\zeta}_R t$ then $s \aux_R t$. By the depth condition
    there is $\gamma<\zeta$ such that for every $\delta \ge \gamma$ the
    redex contracted in~$S$ at~$\delta$ occurs at depth greater than
    zero. Let~$t_\gamma$ be the term at index~$\gamma$ in~$S$. Then by
    the inductive hypothesis we have $s \aux_R t_\gamma$, and thus $s
    \infred_R t_\gamma$ by Lemma~\ref{lem_two_to_one}. There are three
    cases.
    \begin{itemize}
    \item $t_\gamma \equiv a$. This is impossible because then there
      can be no reduction of~$t_\gamma$ at depth greater than zero.
    \item $t_\gamma \equiv \lambda x . r$. Then $t \equiv \lambda x
      . u$ and $r \xto{S',\delta}_R u$ with $\delta \le \zeta$. Hence
      $r \aux_R u$ by the coinductive hypothesis if $\delta=\zeta$,
      or by the inductive hypothesis if $\delta < \zeta$. Since $s
      \infred_R \lambda x . r$ we obtain $s \aux_R \lambda x . u
      \equiv t$.
    \item $t_\gamma \equiv t_1t_2$. Then $t\equiv u_1u_2$ and the tail
      of the reduction~$S$ past~$\gamma$ may be split into two parts:
      $t_i \xto{S_i,\delta_i}_R u_i$ with $\delta_i \le \zeta$ for
      $i=0,1$. Then $t_i \aux_R u_i$ by the inductive and/or the
      coinductive hypothesis. Since $s \infred_R t_1t_2$ we obtain $s
      \aux_R u_1u_2 \equiv t$.\qedhere
    \end{itemize}
  \end{itemize}
\end{proof}

\begin{cor}[$\omega$-compression]
  If~$\to_R$ is appendable and there exists a strongly convergent
  $R$-reduction sequence from~$s$ to~$t$ then there exists such a
  sequence of length at most~$\omega$.
\end{cor}

\begin{cor}
  Let~$\Uc$ be a set of strongly meaningless terms.
  \begin{itemize}
  \item $s \infred_{\beta\bot_\Uc} t$ iff there exists a strongly
    convergent $\beta\bot_\Uc$-reduction sequence from~$s$ to~$t$.
  \item $s \infred_{\beta} t$ iff there exists a strongly convergent
    $\beta$-reduction sequence from~$s$ to~$t$.
  \end{itemize}
\end{cor}

\begin{proof}
  By Theorem~\ref{thm_strongly_convergent} it suffices to show
  that~$\to_{\beta\bot_\Uc}$ and~$\to_\beta$ are
  appendable. For~$\to_{\beta\bot_\Uc}$ this follows from
  Lemma~\ref{lem_omega_merge} and
  Corollary~\ref{cor_bohm_beta_append}. For~$\to_\beta$ this follows
  from Lemma~\ref{lem_beta_beta_fin_append}.
\end{proof}

\section{The formalisation}\label{sec_formalisation}

The results of this paper have been formalised in the Coq proof
assistant. The formalisation is available at:
\begin{center}
  \url{https://github.com/lukaszcz/infinitary-confluence}
\end{center}
The formalisation contains all results of
Section~\ref{sec_inflam_confluence}. We did not formalise the proof
from Section~\ref{sec_inflam_strong_convergence} of the equivalence
between the coinductive definition of the infinitary reduction
relation and the standard notion of strongly convergent reductions.

In our formalisation we use a representation of infinitary
lambda terms with de Bruijn indices, and we do not allow constants
except~$\bot$. Hence, the results about $\alpha$-conversion alluded to
in Section~\ref{sec_inflam_intro} are not formalised either. Because
the formalisation is based on de Bruijn indices, many tedious lifting
lemmas need to be proved. These lemmas are present only in the
formalisation, but not in the paper.

In general, the formalisation follows closely the text of the
paper. Each lemma from Section~\ref{sec_inflam_confluence} has a
corresponding statement in the formalisation (annotated with the lemma
number from the paper). There are, however, some subtleties, described
below.

One difficulty with a Coq formalisation of our results is that in~Coq
the coinductively defined equality (bisimilarity)~$=$ on infinite
terms (see Definition~\ref{def_bisimilarity}) is not identical with
Coq's definitional equality~$\equiv$. In the paper we use~$\equiv$
and~$=$ interchangeably, following
Proposition~\ref{prop_bisimilarity}. In the formalisation we needed to
formulate our definitions ``modulo'' bisimilarity. For instance, the
inductive definition of the transitive-reflexive closure~$R^*$ of a
relation~$R$ on infinite terms is as follows.
\begin{enumerate}
\item If $t_1 = t_2$ then $R^* t_1 t_2$ (where~$=$ denotes
  bisimilarity coinductively defined like in
  Definition~\ref{def_bisimilarity}).
\item If $R t_1 t_2$ and $R^* t_2 t_3$ then $R^* t_1 t_3$.
\end{enumerate}
Changing the first point to
\begin{enumerate}
\item $R^* t t$ for any term~$t$
\end{enumerate}
would not work with our formalisation. Similarly, the formal
definition of the compatible closure of a relation~$R$ follows the
inductive rules
\[
  \begin{array}{cccc}
    \infer{s \to_R t}{\la s, t \ra \in R} &\quad
    \infer{s t \to_R s' t'}{s \to_R s' & t = t'} &\quad
    \infer{s t \to_R s' t'}{t \to_R t' & s = s'} &\quad
    \infer{\lambda x . s \to_R \lambda x . s'}{s \to_R s'}
  \end{array}
\]
where~$=$ denotes the coinductively defined bisimilarity relation.

Another limitation of Coq is that it is not possible to directly prove
by coinduction statements of the form
$\forall \vec{x} . \varphi(\vec{x}) \to R_1(\vec{x}) \land
R_2(\vec{x})$, i.e., statements with a conjuction of two coinductive
predicates. Instead, we show
$\forall \vec{x} . \varphi(\vec{x}) \to R_1(\vec{x})$ and
$\forall \vec{x} . \varphi(\vec{x}) \to R_2(\vec{x})$ separately. In
all our coinductive proofs we use the coinductive hypothesis in a way
that makes this separation possible.

The formalisation assumes the following axioms.
\begin{enumerate}
\item The constructive indefinite description axiom:
  \[
    \forall A : \Type . \forall P : A \to \Prop . (\exists x : A . P
    x) \to \{x : A \mid P x\}.
  \]
  This axiom states that if there exists an object~$x$ of type~$A$
  satisfying the predicate~$P$, then it is possible to choose a fixed
  such object. This is not provable in the standard logic of Coq. We
  need this assumption to be able to define the implicit functions in
  some coinductive proofs which show the existence of an infinite
  object, when the form of this object depends on which case in the
  definition of some (co)inductive predicate holds. More precisely,
  the indefinite description axiom is needed in the proof of
  Lemma~\ref{lem_bohm_decompose}, in the definition of canonical
  root normal forms (Definition~\ref{def_crnf}), and in the proofs of
  Lemma~\ref{lem_sim_beta_2}, Lemma~\ref{lem_sim_to_bot} and
  Lemma~\ref{lem_N_normalising}.
\item Excluded middle for the property of being in root normal form:
  for every term~$t$, either~$t$ is in root normal form or not.
\item Excluded middle for the property of having a root normal form:
  for every term~$t$, either~$t$ has a root normal form or not.
\item Excluded middle for the property of belonging to a set of
  strongly meaningless terms: for any set of strongly meaningless
  terms~$\Uc$ and any term~$t$, either $t \in \Uc$ or $t \notin \Uc$.
\end{enumerate}
Note that the last axiom does not constructively imply the third. We
define being root-active as not having a root normal form. In fact, we
need the third axiom to show that if a term does not belong to a set
of meaningless terms then it has a root normal form.

The first axiom could probably be avoided by making the reduction
relations $\Set$-valued instead of $\Prop$-valued. We do not use the
impredicativity of~$\Prop$. The reason why we chose to define the
relations as $\Prop$-valued is that certain automated proof search
tactics work better with $\Prop$-valued relations, which makes the
formalisation easier to carry out.

Because the $\bot_\Uc$-reduction rule, for any set of meaningless
terms~$\Uc$, requires an oracle to check whether it is applicable, the
present setup is inherently classical. It is an interesting research
question to devise a constructive theory of meaningless terms.

Aside of the axioms (1)--(4), everything else from
Section~\ref{sec_inflam_confluence} is formalised in the constructive
logic of Coq, including the proof of Theorem~\ref{thm_polonsky} only
cited in this paper. Our formalisation of Theorem~\ref{thm_polonsky}
closely follows~\cite{EndrullisPolonsky2011}.

In our formalisation we extensively used the CoqHammer
tool~\cite{CzajkaKaliszyk2018}.

\section{Conclusions}

We presented new and formal coinductive proofs of the following
results.
\begin{enumerate}
\item Confluence and normalisation of B{\"o}hm reduction over any set
  of strongly meaningless terms.
\item Confluence and normalisation of B{\"o}hm reduction over
  root-active terms, by showing that root-active terms are strongly
  meaningless.
\item Confluence of infinitary $\beta$-reduction modulo any set of
  meaningless terms (expansion axiom not needed).
\end{enumerate}
We formalised these results in Coq. Our formalisation uses a
definition of infinitary lambda terms based on de Bruijn
indices. Strictly speaking, the precise relation of this definition to
other definitions of infinitary lambda terms in the literature has not
been established. We leave this for future work. The issue of the
equivalence of various definitions of infinitary lambda terms is not
necessarily
trivial~\cite{KurzPetrisanSeveriVries2012,KurzPetrisanSeveriVries2013}.

By a straightforward generalisation of a result
in~\cite{EndrullisPolonsky2011} we also proved equivalence, in the
sense of existence, of the coinductive definitions of infinitary
rewriting relations with the standard definitions based on strong
convergence. However, we did not formalise this result.

In Section~\ref{sec_coind} we explained how to elaborate our
coinductive proofs by reducing them to proofs by transfinite induction
and thus eliminating coinduction. This provides one way to understand
and verify our proofs without resorting to a formalisation. After
properly understanding the observations of Section~\ref{sec_coind} it
should be ``clear'' that coinduction may in principle be eliminated in
the described manner. We use the word ``clear'' in the same sense that
it is ``clear'' that the more sophisticated inductive constructions
commonly used in the literature can be formalised in ZFC set
theory. Of course, this notion of ``clear'' may always be debated. The
only way to make this completely precise is to create a formal system
based on Section~\ref{sec_coind} in which our proofs could be
interpreted reasonably directly. We do not consider the observations
of Section~\ref{sec_coind} to be novel or particularly
insightful. However, distilling them into a formal system could
perhaps arise some interest. This is left for future work.

\bibliography{biblio}{}
\bibliographystyle{plain}

\appendix

\clearpage
\section{A set of meaningless terms not satisfying the expansion axiom}

We recall the definitions from page~14. Let~$\Os$ be the \emph{ogre}
satisfying $\Os \equiv \lambda x . \Os$. A term~$t$ is
\emph{head-active} if $t \equiv \lambda x_1 \ldots x_n . r t_1 \ldots
t_m$ with $r \in \Rc$ and $n,m\ge 0$. Define:
\begin{itemize}
\item $\Hc = \{ t \in \Lambda^\infty \mid t \reduces_\beta t' \text{
  with } t' \text{ head-active} \}$,
\item $\Oc = \{ t \in \Lambda^\infty \mid t \reduces_\beta \Os \}$,
\item $\Uc = \Hc \cup \Oc$.
\end{itemize}
We will show that~$\Uc$ is a set of meaningless terms. The proofs in
this appendix rely on the results established in the paper. In
particular, we use that~$\Rc$ is a set of strongly meaningless terms
and that~$\infred_\beta$ is confluent modulo~$\Rc$.

\begin{lem}\label{lem_a_cong}
  \begin{enumerate}
  \item If $t \in \Uc$ then $t s \in \Uc$.
  \item If $t \in \Uc$ then $\lambda x . t \in \Uc$.
  \item If $t \in \Uc$ then $t[s/x] \in \Uc$.
  \end{enumerate}
\end{lem}

\begin{proof}
  Follows from definitions and the fact that $\Rc$ satisfies the
  substitution axiom.
\end{proof}

\begin{lem}
  If $t \infred_\beta t'$ and $t \in \Hc$ then $t' \in \Hc$.
\end{lem}

\begin{proof}
  We have $t \reduces_\beta u$ with $u$ head-active. By confluence
  modulo~$\Rc$ there are~$s,s'$ with $u \infred_\beta s$, $t'
  \infred_\beta s'$ and $s \sim_\Rc s'$. It follows from the closure
  and indiscernibility axioms for~$\Rc$ that $s'$ is head-active. Now,
  using the expansion axiom for~$\Rc$ one may show that there is a
  head-active $w$ such that $t' \reduces_\beta w \infred_\beta s'$.
\end{proof}

\begin{lem}
  If $t \infred_\beta t'$ and $t \in \Oc$ then $t' \in \Oc$.
\end{lem}

\begin{proof}
  We have $t \reduces_\beta \Os$. Because the reduction is finite, no
  $\beta$-contractions occur below a fixed depth~$d$. We can write
  $\Os \equiv \lambda \vec{x} . \Os$ where on the right side~$\Os$
  occurs below depth~$d$. Then there is~$u$ with exactly one
  occurrence of~$z$ (where $z \notin \vec{x}$) such that $u
  \reduces_\beta \lambda \vec{x} . z$ and $t \equiv u[\Os/z]$.

  Write $s \succ_z^\Os s'$ if~$s'$ may be obtained from~$s$ by
  changing some~$\Os$ subterms in~$s$ into some terms having the form
  $zu_1\ldots u_n$, defined coinductively analogously to
  Definition~\ref{def_succ}. One shows:
  \begin{enumerate}
  \item if $s \succ_z^\Os s'$ and $s \reduces_\beta w$ then there
    is~$w'$ with $s' \reduces_\beta w'$ and $w \succ_z^\Os w'$,
  \item if $s \succ_z^\Os s'$ and $s' \reduces_\beta w'$ then there
    is~$w$ with $s \reduces_\beta w$ and $w \succ_z^\Os w'$,
  \item if $s \succ_z^\Os s'$ and $s \infred_\beta w$ then there
    is~$w'$ with $s' \infred_\beta w'$ and $w \succ_z^\Os w'$,
  \end{enumerate}
  The first two points are proved by induction. The third one follows
  from the first one using coinduction.

  Hence, there exists~$u'$ such that $u \infred_\beta u'$, and $t'
  \succ_z^\Os u'$. By confluence modulo~$\Rc$ and the fact that
  $\lambda \vec{x} . z$ is a finite normal form, we have $u'
  \reduces_\beta \lambda \vec{x} . z$. Thus there is $w$ with $t'
  \reduces_\beta w \succ_z^\Os \lambda \vec{x} . z$. This is possible
  only when $w \equiv \Os$.
\end{proof}

\begin{cor}\label{cor_a_closure}
  If $t \infred_\beta t'$ and $t \in \Uc$ then $t' \in \Uc$.
\end{cor}

\begin{lem}\label{lem_a_sim_subst}
  If $t_1 \sim_\Uc s_1$ and $t_2 \sim_\Uc s_2$ then $t_1[t_2/x]
  \sim_\Uc s_1[s_2/x]$.
\end{lem}

\begin{proof}
  By coinduction, using Lemma~\ref{lem_a_cong}(3).
\end{proof}

\begin{lem}\label{lem_a_sim_preserve}
  If $t \to_\beta t'$ and $t \sim_\Uc s$ then there is~$s'$ with $s
  \to_\beta^\equiv s'$ and $t' \sim_\Uc s'$.
\end{lem}

\begin{proof}
  Induction on~$t \to_\beta t'$. If the case $t,s \in \Uc$ in the
  definition of~$t \sim_\Uc s$ holds, then $t' \in \Uc$ by
  Corollary~\ref{cor_a_closure}, so $t' \sim_\Uc s$ and we may take
  $s' \equiv s$. So assume otherwise. Then all cases follow directly
  from the inductive hypothesis except when~$t$ is the contracted
  $\beta$-redex. Then $t \equiv (\lambda x . t_1) t_2$ and $t' \equiv
  t_1[t_2/x]$ and $s = s_0 s_2$ with $\lambda x . t_1 \sim_\Uc s_0$
  and $t_2 \sim_\Uc s_2$. If $\lambda x . t_1, s_0 \in \Uc$ then
  $t,s\in\Uc$ by Lemma~\ref{lem_a_cong}, and thus $t'\in\Uc$ by
  Corollary~\ref{cor_a_closure}, and thus $t' \sim_\Uc s$ and we may
  take $s' \equiv s$. Otherwise, $s_0 \equiv \lambda x . s_1$ with
  $t_1 \sim_\Uc s_1$. By Lemma~\ref{lem_a_sim_subst} we have $s
  \to_\beta s_1[s_2/x] \sim_\Uc t_1[t_2/x] \equiv t'$, so we take $s'
  \equiv s_1[s_2/x]$.
\end{proof}

\begin{lem}\label{lem_a_sim_rc}
  If $t \in \Rc$ and $t \sim_\Uc t'$ then $t' \in \Uc$.
\end{lem}

\begin{proof}
  Assume $t' \infred_\beta s_0$ with $s_0$ in rnf. Then by
  Lemma~\ref{lem_beta_to_rnf} there is~$s'$ in rnf with $t'
  \reduces_\beta s'$. By Lemma~\ref{lem_a_sim_preserve} there is~$s$
  with $t \reduces_\beta s$ and $s \sim_\Uc s'$. If $s,s' \in \Uc$
  then $t' \in \Uc$. Otherwise, because~$t \in \Rc$, we must have $s
  \equiv s_1 s_2$, $s' \equiv s_1's_2'$ with $s_i \sim_\Uc
  s_i'$. Because $t \in \Rc$ and $t \reduces_\beta s_1 s_2$, there
  exists $u$ such that $s_1 \reduces_\beta \lambda x . u$. Then by
  Lemma~\ref{lem_a_sim_preserve} there is $u_0$ with $s_1'
  \reduces_\beta u_0 \sim_\Uc \lambda x . u$. If $u_0,\lambda . u \in
  \Uc$ then $s_1' \in \Uc$, and thus $s' \in \Uc$, and thus $t' \in
  \Uc$. Otherwise $u_0 \equiv \lambda x . u'$ with $u \sim_\Uc
  u'$. But this contradicts that~$s'$ is in rnf.
\end{proof}

\begin{lem}\label{lem_a_sim}
  If $t \in \Uc$ and $t \sim_\Uc t'$ then $t' \in \Uc$.
\end{lem}

\begin{proof}
  First assume $t \in \Hc$, i.e., $t \reduces_\beta u \equiv \lambda
  x_1 \ldots x_n . r t_1 \ldots t_m$ with $r \in \Rc$. By
  Lemma~\ref{lem_a_sim_preserve} there is~$u$ with $t' \reduces_\beta
  u' \sim_\Uc u$. We may assume $u' \equiv \lambda x_1 \ldots x_n . r'
  t_1' \ldots t_m'$ with $r' \sim_\Uc r$ (otherwise $u' \in \Uc$,
  using Lemma~\ref{lem_a_cong}, so $t' \in \Uc$). But then $r' \in
  \Rc$ by Lemma~\ref{lem_a_sim_rc}. Hence $t' \in \Hc \subseteq \Uc$.

  Now assume $t \in \Oc$, i.e., $t \reduces_\beta \Os$. By
  Lemma~\ref{lem_a_sim_preserve} there is~$u'$ with $t' \reduces_\beta
  u' \sim_\Uc \Os$. Using Lemma~\ref{lem_a_cong} one checks that $u'
  \sim_\Uc \Os$ implies $u' \in \Uc$. Then also $t' \in \Uc$.
\end{proof}

\begin{thm}
  $\Uc$ is a set of meaningless terms.
\end{thm}

\begin{proof}
  The closure axiom follows from Corollary~\ref{cor_a_closure}. The
  substitution axiom follows from Lemma~\ref{lem_a_cong}(3). The
  overlap axiom follows from Lemma~\ref{lem_a_cong}(1). The
  root-activeness axiom follows from $\Rc \subseteq \Hc \subseteq
  \Uc$. The indiscernibility axiom follows from Lemma~\ref{lem_a_sim}.
\end{proof}

\end{document}